\documentclass[journal]{IEEEtran}

\usepackage{cite}
\usepackage[cmex10]{amsmath}
\usepackage{amssymb}
\usepackage{amsthm}
\usepackage{mathrsfs}
\usepackage{bm}
\usepackage[mathscr]{eucal}
\usepackage{amssymb,amsmath,amsthm,amsfonts,latexsym}
\usepackage{amsmath,graphicx,bm,xcolor,url}
\usepackage{graphicx}
\usepackage{latexsym}
\usepackage{CJK}
\usepackage{indentfirst}
\usepackage{geometry}
\usepackage{psfrag}
\usepackage{setspace}
\usepackage{algorithmic}
\usepackage{algorithmic, cite}
\usepackage{algorithm}
\usepackage{array}
\usepackage{mdwmath}
\usepackage{mdwtab}
\usepackage{eqparbox}
\usepackage{url}
\usepackage{epstopdf}
\usepackage{epsfig,epsf,psfrag}
\usepackage{fixltx2e}
\usepackage{verbatim}
\usepackage{textcomp}
\usepackage{psfrag} 

\usepackage{subfigure} 
\usepackage{caption}

\theoremstyle{plain}
\newtheorem{mythe}{Theorem}
\theoremstyle{remark}

\theoremstyle{plain}

\theoremstyle{remark}
\newtheorem{mypro}{Proposition}
\theoremstyle{plain}

\theoremstyle{remark}
\newtheorem{myrem}{Remark}
\theoremstyle{remark}

\theoremstyle{remark}

\theoremstyle{remark}

\theoremstyle{remark}

\theoremstyle{remark}

\catcode`~=11 \def\UrlSpecials{\do\~{\kern -.15em\lower .7ex\hbox{~}\kern .04em}} \catcode`~=13

\allowdisplaybreaks[4]

\newcommand{\calA}{\mathcal{A}}

\newcommand{\calC}{\mathcal{C}}

\newcommand{\calH}{\mathcal{H}}

\newcommand{\calN}{\mathcal{N}}

\newcommand{\calQ}{\mathcal{Q}}

\newcommand{\bA}{\mathbf{A}}

\newcommand{\bg}{\mathbf{g}}

\newcommand{\bh}{\mathbf{h}}
\newcommand{\bH}{\mathbf{H}}

\newcommand{\bI}{\mathbf{I}}

\newcommand{\bs}{\mathbf{s}}

\newcommand{\bt}{\mathbf{t}}
\newcommand{\bT}{\mathbf{T}}
\newcommand{\bu}{\mathbf{u}}
\newcommand{\bU}{\mathbf{U}}
\newcommand{\bv}{\mathbf{v}}
\newcommand{\bV}{\mathbf{V}}
\newcommand{\bw}{\mathbf{w}}

\newcommand{\bx}{\mathbf{x}}

\newcommand{\by}{\mathbf{y}}

\newcommand{\bz}{\mathbf{z}}

\newcommand{\bbE}{\mathbb{E}}

\DeclareMathAlphabet{\mathbsf}{OT1}{cmss}{bx}{n}
\DeclareMathAlphabet{\mathssf}{OT1}{cmss}{m}{sl}

\DeclareSymbolFont{bsfletters}{OT1}{cmss}{bx}{n}
\DeclareSymbolFont{ssfletters}{OT1}{cmss}{m}{n}
\DeclareMathSymbol{\bsfGamma}{0}{bsfletters}{'000}
\DeclareMathSymbol{\ssfGamma}{0}{ssfletters}{'000}
\DeclareMathSymbol{\bsfDelta}{0}{bsfletters}{'001}
\DeclareMathSymbol{\ssfDelta}{0}{ssfletters}{'001}
\DeclareMathSymbol{\bsfTheta}{0}{bsfletters}{'002}
\DeclareMathSymbol{\ssfTheta}{0}{ssfletters}{'002}
\DeclareMathSymbol{\bsfLambda}{0}{bsfletters}{'003}
\DeclareMathSymbol{\ssfLambda}{0}{ssfletters}{'003}
\DeclareMathSymbol{\bsfXi}{0}{bsfletters}{'004}
\DeclareMathSymbol{\ssfXi}{0}{ssfletters}{'004}
\DeclareMathSymbol{\bsfPi}{0}{bsfletters}{'005}
\DeclareMathSymbol{\ssfPi}{0}{ssfletters}{'005}
\DeclareMathSymbol{\bsfSigma}{0}{bsfletters}{'006}
\DeclareMathSymbol{\ssfSigma}{0}{ssfletters}{'006}
\DeclareMathSymbol{\bsfUpsilon}{0}{bsfletters}{'007}
\DeclareMathSymbol{\ssfUpsilon}{0}{ssfletters}{'007}
\DeclareMathSymbol{\bsfPhi}{0}{bsfletters}{'010}
\DeclareMathSymbol{\ssfPhi}{0}{ssfletters}{'010}
\DeclareMathSymbol{\bsfPsi}{0}{bsfletters}{'011}
\DeclareMathSymbol{\ssfPsi}{0}{ssfletters}{'011}
\DeclareMathSymbol{\bsfOmega}{0}{bsfletters}{'012}
\DeclareMathSymbol{\ssfOmega}{0}{ssfletters}{'012}

\newcommand{\tila}{\widetilde{a}}

\newcommand{\hatc}{\widehat{c}}

\newcommand{\tilH}{\widetilde{H}}

\newcommand{\tilbh}{\widetilde{\bh}}
\newcommand{\tilbH}{\widetilde{\bH}}

\newcommand{\tilP}{\widetilde{P}}

\newcommand{\hats}{\widehat{s}}

\newcommand{\tils}{\widetilde{s}}

\newcommand{\tilbT}{\widetilde{\bT}}

\newcommand{\tilu}{\widetilde{u}}

\newcommand{\tilbu}{\widetilde{\bu}}

\newcommand{\tilx}{\widetilde{x}}

\newcommand{\hatbx}{\widehat{\bx}}

\newcommand{\tilbx}{\widetilde{\bx}}

\newcommand{\tily}{\widetilde{y}}

\newcommand{\tilby}{\widetilde{\by}}

\newcommand{\tilbz}{\widetilde{\bz}}

\newcommand{\barx}{\bar{x}}

\newcommand{\barP}{\bar{P}}

\newcommand{\bLambda}{\bm{\Lambda}}

\def\norm#1{\left\| #1 \right\|}
\def\norm2#1{\left\| #1 \right\|_2}
\def\norm22#1{\left\| #1 \right\|_2^2}

\newcommand{\eqa}{\stackrel{(a)}{=}}

\DeclareMathOperator{\diag}{diag}

\newcommand{\qednew}{\nobreak \ifvmode \relax \else
      \ifdim\lastskip<1.5em \hskip-\lastskip
      \hskip1.5em plus0em minus0.5em \fi \nobreak
      \vrule height0.75em width0.5em depth0.25em\fi}

\usepackage{subfigure} 
\usepackage{dblfloatfix}
\usepackage{fixltx2e}

\newcommand{\black}[1]{{{\color{black} #1}}}

\title{Cooperative Ambient Backscatter Communications for Green Internet-of-Things} 

\author{Gang~Yang, \emph{Member, IEEE}, Qianqian~Zhang, and Ying-Chang~Liang, \emph{Fellow, IEEE}
\thanks{G.~Yang and Q. Zhang are with the National Key Laboratory of Science and Technology on Communications, University of Electronic Science and Technology of China (UESTC), Chengdu, P. R. China (e-mails: yanggang@uestc.edu.cn; qqzhang\_kite@163.com).}
\thanks{Y.-C.~Liang is with the Center for Intelligent Networking and Communications, University of Electronic Science and Technology of China (UESTC), Chengdu, P. R. China (e-mail: liangyc@ieee.org).}
\thanks{A preliminary version of this paper was presented in \cite{YangLiangZhangICC17}.}}

\begin{document}
\maketitle 

\begin{abstract}
Ambient backscatter communication (AmBC) enables a passive backscatter device to transmit information to a reader using ambient RF signals, and has emerged as a promising solution to green Internet-of-Things (IoT). Conventional AmBC receivers are interested in recovering the information from the ambient backscatter device (A-BD) only. In this paper, we propose a cooperative AmBC (CABC) system in which the reader recovers information not only from the A-BD, but also from the RF source. We first establish the system model for the CABC system from spread spectrum and spectrum sharing perspectives. Then, for flat fading channels, we derive the optimal maximum-likelihood (ML) detector, suboptimal linear detectors as well as successive interference-cancellation (SIC) based detectors. For frequency-selective fading channels, the system model for the CABC system over ambient orthogonal frequency division multiplexing (OFDM) carriers is proposed, upon which a low-complexity optimal ML detector is derived. For both kinds of channels, the bit-error-rate (BER) expressions for the proposed detectors are derived in closed forms. Finally, extensive numerical results have shown that, when the A-BD signal and the RF-source signal have equal symbol period, the proposed SIC-based detectors can achieve near-ML detection performance for typical application scenarios, and when the A-BD symbol period is longer than the RF-source symbol period, the existence of backscattered signal in the CABC system can enhance the ML detection performance of the RF-source signal, thanks to the beneficial effect of the backscatter link when the A-BD transmits at a lower rate than the RF source.
\end{abstract}


\vspace{-0.2cm}
\section{Introduction}\label{introduction}
Ambient backscatter communication (AmBC) enables ambient backscatter devices (A-BDs) to modulate their information symbols over the ambient RF signals (e.g., cellular, TV or WiFi signals) without using a complex RF transmitter~\cite{ABCSigcom13}. On the other hand, compared to traditional backscatter communication systems such as radio-frequency identification (RFID) systems~\cite{Dobkinbook2007, BoyerSumit14}, AmBC does not require the reader to transmit a high power RF sinusoidal carrier to the backscatter device. Thus, AmBC is a promising solution to Internet-of-Things (IoT)~\cite{Ishizaki11} with stringent cost, power, and complexity constraints, and has drawn significant interest from both academia and industry recently.

One of the key challenges in the receiver design for AmBC is to tackle the direct-link interference from the ambient RF source. Some existing methods treat the direct-link interference as part of the background noise~\cite{ABCSigcom13, WiFiBackscatter14,QianGaoAmBCTWC16, WangGaoAmBCTCOM16}. In \cite{ABCSigcom13} and \cite{WiFiBackscatter14}, energy detectors are used to detect the A-BD symbols. In~\cite{QianGaoAmBCTWC16} and~\cite{WangGaoAmBCTCOM16}, maximum-likelihood (ML) detection is proposed for differential modulation. Because of the double-attenuation in the backscatter link, the above proposed detection schemes suffer from severe performance degradation due to the strong direct-link interference. Recently, interference cancellation techniques have been applied to the receiver design for AmBC~\cite{YangLiangGC16, YangLiangZhangPei17, TurbochargingABCSigcom14, BackFiSigcom15,HitchHikeKattiSenSys16}. In~\cite{YangLiangGC16} and~\cite{YangLiangZhangPei17}, the direct-link interference is cancelled out by exploiting the repeating structure of the ambient orthogonal frequency division multiplexing (OFDM) signals. In~\cite{TurbochargingABCSigcom14}, two receive antennas are used at the reader to cancel out the effect of the RF-source signal by calculating the ratio of the amplitudes of the signals received at the two antennas. In \cite{BackFiSigcom15}, a WiFi backscatter system is proposed in which the WiFi access point (AP) detects the received signal backscattered from the A-BD while simultaneously transmitting WiFi packages to a standard WiFi client. This design relies on the self-interference cancellation technique developed for full-duplex communication.

There are other studies on AmBC addressing the problem of direct-link interference~\cite{PassiveWiFiNSDI16,InterscatterSigcom16, FSBackscatterSigcomm16}. A passive WiFi system is proposed in~\cite{PassiveWiFiNSDI16} which requires a dedicated device to transmit RF sinusoidal carrier at a frequency that lies outside the desired WiFi channel, such that the WiFi receiver can suppress the resulting out-of-band (direct-link) carrier interference. An inter-technology backscatter system is proposed in~\cite{InterscatterSigcom16}, which transforms wireless transmissions from one technology (e.g., Bluetooth) to another (e.g., WiFi) in the air. The A-BD creates frequency shifts on a single side of the carrier by using complex impedance of its backscatter circuit, so as to suppress the direct-link interference. A frequency-shifted backscatter (FS-Backscatter) system is proposed in~\cite{FSBackscatterSigcomm16} for on-body sensor applications, which suppresses the direct-link interference by shifting the backscattered signal to a clean band that does not overlap with the direct-link signal. Similarly, an FM backscatter system is proposed in~\cite{WangSmithFMBackscatter17} which uses ubiquitous FM signals as RF source and shifts the backscattered signal to another unoccupied FM channel to suppresses the direct-link interference. However, additional spectrum is needed for the above AmBC systems.



Existing studies focus on designing receivers to recover only the information from the A-BD, while treating the signal from the RF source as unwanted interference. For some AmBC systems like WiFi Backscatter~\cite{WiFiBackscatter14}, BackFi~\cite{BackFiSigcom15} and HitchHike~\cite{HitchHikeKattiSenSys16}, the A-BD information is recovered by commodity devices such as smart phones and WiFi APs. In fact, the commodity device can simultaneously recover the A-BD information through backscatter communication, when it is receiving information from the RF source. {{\textcolor{black}{Two typical application examples are described as follows, a smartphone simultaneously recovers information
either from both a home WiFi AP and a domestic sensor for smart-home applications, or from both a cellular base station and a wearable sensor for body-area-network applications.}}} In this paper, we propose a cooperative AmBC (CABC) system, for which a novel receiver, called cooperative receiver (C-RX), is designed to recover information from both the RF source and the A-BD. \black{The cooperation in the CABC system also lies in the fact that the backscattering A-BD acts as a (passive) relay to assist the recovery of RF-source information at the C-RX.} We are interested in the receiver design and performance analysis for such CABC system. The main contributions of this paper are summarized as follows:
\begin{itemize}
  \item First, we establish the system model for the CABC system with multiple receive antennas under flat fading channels. Since the received signal at the C-RX contains the direct-link signal from the ambient RF source and the backscatter-link signal from the A-BD, and the backscatter-link signal is the multiplication of the RF-source signal and the A-BD signal, both spread spectrum and spectrum sharing perspectives are incorporated into the system modelling.
  \item Then, the optimal ML detector is proposed for the C-RX of CABC system. We also propose suboptimal linear detectors and successive interference-cancellation (SIC) based detectors, by exploiting the structural property of the system model. For SIC-based detectors, the C-RX first detects the RF-source signal, then subtracts its resultant direct-link interference from the received signal, and recovers the A-BD signal. Finally, based on the recovered A-BD signal, the C-RX re-estimates the RF-source signal.
  \item We also investigate the receiver design for the CABC system over ambient OFDM carriers under frequency-selective fading channels. We choose the A-BD symbol period to be the same as the OFDM symbol period, and develop a low-complexity optimal ML detector for the C-RX.
  \item We obtain the bit-error-rate (BER) expressions in closed forms for the proposed detectors, under both flat fading channels and frequency-selective fading channels.
  \item Extensive numerical results have shown that when the A-BD signal and the RF-source signal have equal symbol period, the proposed SIC-based detector can achieve near-ML detection performance when the backscattered signal power is lower than the direct-link signal power. When the A-BD symbol period is longer than the RF-source symbol period, the existence of backscattered signal in the CABC system can significantly enhance the ML detection performance of the RF-source signal, compared to conventional single-input-multiple-output (SIMO) communication systems without an A-BD, thanks to the beneficial effect of the backscatter link when the A-BD transmits at a low rate than the ambient RF source. Also, for frequency-selective fading channels, the proposed detector is shown to be robust against the typically very small time delay between the arrival of direct-link signal and the backscatter-link signal at the C-RX.
\end{itemize}

The rest of this paper is organized as follows: Section~\ref{sec:model_flat} establishes the system model for the CABC system under flat fading channels. Section~\ref{sec:receiver_flat} derives the optimal ML detector, linear detectors and SIC-based detectors for CABC under flat fading channels. Section~\ref{sec:OFDM} first establishes the system model for CABC system over ambient OFDM carriers under frequency-selective fading channels, then derives the low-complexity optimal ML detector. Section~\ref{sec:ber_analysis} analyzes the BER performance of CABC systems with various proposed detectors. Section~\ref{sec:simulation} provides numerical results which evaluate the performance of the proposed detectors. Finally, Section~\ref{conslusion} concludes this paper.

The main notations in this paper are listed as follows: The lowercase, boldface lowercase, and boldface uppercase letter  $t$, $\bt$, and $\bT$ denotes a scalar variable (or constant), vector, and matrix, respectively. $|t|$ means the operation of taking the absolute value. $\|\bt\|$ denotes the norm of vector $\bt$. $\calC \calN(\mu, \sigma^2)$ denotes the circularly symmetric complex Gaussian (CSCG) distribution with mean $\mu$ and variance $\sigma^2$. $\bbE[t]$ denotes the statistical expectation of $t$. $t^{\ast}$ means the conjugate of $t$. $\bT^T$ and $\bT^H$ denotes the transpose and conjugate transpose of the matrix $\bT$, respectively. $\textrm{Re}\{ \cdot\}$ and $\textrm{Im}\{ \cdot\}$ denotes the real-part operation and the imaginary-part operation, respectively.

\section{System Model For CABC System Under Flat Fading Channels}\label{sec:model_flat}
In this section, we first describe the proposed CABC system, then establish its system model under flat fading channels.

\vspace{-0.2cm}
\subsection{System Description} \label{SystemDescription}
Fig.~\ref{fig:Fig1} illustrates the system model of the proposed CABC system, which consists of a single-antenna RF source (e.g., TV tower, WiFi AP), a single-antenna A-BD, and a C-RX equipped with $M \ (M \geq 1)$ antennas. The A-BD transmits its modulated signals to the C-RX over the ambient RF carrier. \black{The proposed CABC system is termed as a cooperative system, due to two facts: (1)  the C-RX needs to recover the information from two users, i.e., both the RF source and the A-BD, (2) the backscattering A-BD acts as a (passive) relay to assist the detection of RF-source signal at the C-RX, which will be verified in the sequel of this paper.}

The A-BD contains a single backscatter antenna, a backscatter transmitter (i.e., a switched load impedance), a micro-controller, a memory, a rechargeable battery replenished by an energy harvester, and a signal processor. The energy harvester collects energy from ambient signals and uses it to replenish the battery which provides power for all modules of the A-BD. To transmit information bits stored in the memory to the C-RX, the A-BD modulates its received ambient carrier by intentionally switching the load impedance to change the amplitude and/or phase of its backscattered signal, and the backscattered signal is received and finally decoded by the C-RX. Also, the A-BD antenna can be switched to the signal processor which is able to perform information decoding and other simple signal processing operations such as sensing and synchronization.

\begin{figure}[!t]
\centering
\includegraphics[width=.99\columnwidth] {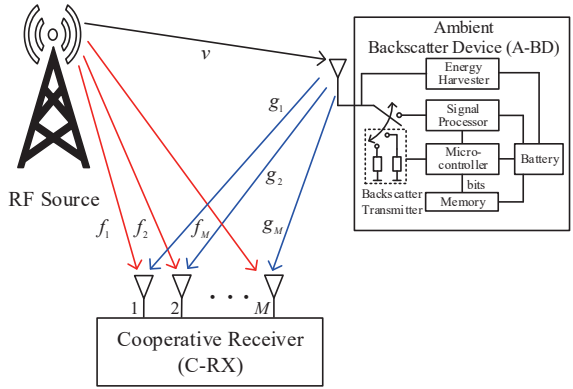}
\caption{System model for a cooperative AmBC system.}
\label{fig:Fig1}
\end{figure}

\vspace{-0.2cm}
\subsection{Signal Model} \label{SignalModel_flat}
In this subsection, we establish the signal model for the proposed CABC system under flat fading channels. The signal model under frequency-selective channels will be discussed in Section~\ref{subsec:signal_ofdm}.


{{\textcolor{black}{Let $R_{\sf s}$ and $R_{\sf c}$ be the symbol rates for the RF-source signal and A-BD signal, respectively. Without loss of generality, we assume $R_{\sf s}= K R_{\sf c}$ with $K$ being a positive integer}}}, since the A-BD data rate is typically smaller than the source data rate~\cite{ABCSigcom13, WiFiBackscatter14,BackFiSigcom15,QianGaoAmBCTWC16}. That is to say, the A-BD symbol $c(n)$ covers $K$ RF-source symbols, denoted as ${\mathbf s} (n) = [s_0(n),\ldots, s_{K-1}(n)]^T \in \calC^{K}$. Let $\calA_{\sf s}$ and $\calA_{\sf c}$ be the modulation alphabet sets of the RF source and the A-BD, respectively. We assume that the A-BD symbol period is smaller than the coherence time of fading channels, and the A-BD symbols synchronize with the RF-source symbols. {{\textcolor{black}{Denote the RF-source's symbol period $T_{\sf s}=1/R_{\sf s}$ and the A-BD's symbol period $T_{\sf c}=1/R_{\sf c}$.}}}

Block fading channel models are assumed. As shown in Fig.~\ref{fig:Fig1}, for the interested block, denote $v$ as the channel coefficient between the RF source and the A-BD, $g_m$ as the channel coefficient between the A-BD and the $m$-th receive antenna, for $m=1, \ \ldots, \ M$, at the C-RX, and $f_m$ as the channel coefficient between the RF source and the $m$-th receive antenna at the C-RX, respectively. We also denote $\bh_1 =[f_1, \ \ldots, \ f_M]^T$ and $\bg =[g_1, \ \ldots, \ g_M]^T$.


Denote the average transmit power at the RF source as $P_{\sf s}$. Let $\alpha$ be the reflection coefficient of the A-BD, \black{which is typically a small (complex) number with absolute value less than 1}, and $c(t)$ be the baseband signal of the A-BD. The backscattered signal\footnote{\black{From the antenna scatterer theorem, the EM field backscattered from the antenna of the A-BD consists of antenna-mode scattering component which relates to re-radiation of closed-circuited antenna and depends on the chip impedance of the A-BD, and the structure-mode scattering one which relates to the scattering from an open-circuited antenna and is load-independent~\cite{Dobkinbook2007}.}} out of the A-BD in baseband form is $\alpha v \sqrt{P_{\sf s}} s(t) c(t)$. \black{Based on such a model for the backscattered signal, the AmBC receivers are implemented in literature~\cite{ABCSigcom13, TurbochargingABCSigcom14,BackFiSigcom15}, and such model is also adopted in recent theoretical work on AmBC~\cite{QianGaoAmBCTWC16,WangGaoAmBCTCOM16,DarsenaVerdeTCOM17}.} Such operation in the A-BD is termed ``modulation in the air'' in~\cite{YangLiangGC16}.

In the $n$-th A-BD symbol period, for $n = 1, 2, \ldots$, the signal received at the $m$-th antenna of the C-RX can be written as
\begin{align}\label{eq:rx_sig}
  y_{m, k} (n) \!=\! f_m \! \sqrt{P_{\sf s}} s_k(n) \!+\! \alpha v g_m \! \sqrt{P_{\sf s}} s_k(n) c(n) \!+\! u_{m, k}(n),
\end{align}
for $m = 1, \ldots, M$, $k=0,\ldots,K-1$, where $u_{m, k}(n) \sim \calC \calN(0, \sigma^2)$. It is assumed that the noises $u_{m, k} (n)$'s are independent of the signals $s_k(n)$'s and $c(n)$'s.

\begin{myrem}
\black{Strictly speaking, the arrival of the backscatter-link signal from the A-BD (i.e., the second term in~\eqref{eq:rx_sig}) at the C-RX is typically delayed by a time $\tau$ ($\tau \geq 0$), compared to the arrival of the direct-link signal from the RF source (i.e., the first term in~\eqref{eq:rx_sig}). However, such a delay is typically negligible in most application scenarios, because of the following facts: (i) the A-BD transmits information to nearby C-RX, and the typical A-BD-to-C-RX distance is less than 10 feet~\cite{ABCSigcom13, WiFiBackscatter14}; (ii) the A-BD symbol period is typically much longer than the propagation delay of the A-BD-to-C-RX channel, since the low-cost and low-power A-BD supports only low-rate backscattering operation~\cite{BackFiSigcom15}. For instance, the propagation delay for a A-BD-to-C-RX distance of 3 meters is 10 ns, and this is much shorter than 1 microsecond which corresponds to a A-BD symbol rate up to 1 Mbps.}
\end{myrem}

\begin{myrem}
\black{For some extreme application scenarios in which the time delay $\tau$ is not negligible, the RF-source signal can be viewed to propagate through a frequency-selective fading channel with two paths equivalently (i.e., the direct-link path $f_m$ and the backscatter-link path $\alpha v g_m$ with additional delay $\tau$). Therefore, the RF source can adopt OFDM modulation to combat the frequency-selective fading, and the C-RX can use the detector proposed in Section~\ref{sec:OFDM} to eliminate the effect of the delay $\tau$. We thus assume that the delay $\tau$ is zero in this section.}
\end{myrem}


For convenience, we define the average receive SNRs of the direct link and the backscatter link as $\gamma_{\sf d} \triangleq \frac{P_{\sf s}\beta_{\sf f}}{\sigma^2}$ and $\gamma_{\sf b} \triangleq \frac{|\alpha|^2 P_{\sf s} \beta_{\sf v} \beta_{\sf g}}{\sigma^2}$, respectively, {where ${\beta_{\sf f}} = \bbE[|f_m|^2]$, ${\beta_{\sf v}} = \bbE[|v|^2]$, and ${\beta_{\sf g}} = \bbE[|g_m|^2]$, $\forall m, k$. We also define the relative SNR between the backscatter link and the direct link as $\Delta \gamma \triangleq \frac{\gamma_{\sf b}}{\gamma_{\sf d}}$.

{{\textcolor{black}{For notational simplicity}}}, we assume that $\alpha=1$, $P_{\sf s} =1$, $\beta_{\sf f}=1$, and vary $\sigma^2$ according to $\gamma_{\sf d}$. {{\textcolor{black}{Notice that this assumption does not affect the analyses and results in the reminder of this paper, since the effect of those constant parameters can be incorporated in the channel coefficients.}}}

Denote $\by_k (n)=[y_{1, k} (n), \ \ldots, \ y_{M, k} (n)]^T$, $\bx_k (n)=[ s_k(n), \ s_k(n) c(n)]^T$, $\bu_k(n)=[u_{1, k} (n), \ldots,$ $u_{M,k} (n)]^T$, and $\bH = [ \bh_1, \bh_2]$, where $\bh_2= \alpha v \bg$. Note $\bh_1$ and $\bh_2$ are the channel responses for the direct link and backscatter link, respectively.

The signal model in~\eqref{eq:rx_sig} can be rewritten as
\begin{align}
  \by_k (n)&= \bh_1 s_k(n) + \bh_2 s_k(n) c(n) + \bu_k (n) \label{eq:yn_1} \\
  &= \bH \bx_k (n) + \bu_k (n).\label{eq:rx_sig_mtx}
\end{align}

From~\eqref{eq:yn_1}, it is seen that the signal backscattered by the A-BD is the multiplication of a low-rate A-BD signal $c(n)$ and the high-rate RF source signal $s_k(n)$. Such operation can be viewed as ``spreading over-the-air'', and the corresponding spreading gain is $K$.

The objective of the C-RX is to recover both the RF-source signal $\bs(n)$ and the A-BD signal $c(n)$ from $\by_k (n)$'s, assuming that the composite channel $\bH$ is known by the C-RX.
\black{Notice that both the direct-link channels $h_{1m}$'s and the composite channels $h_{2m}$'s can be estimated through using pilot signals\footnote{\black{The composite backscatter-link channel was estimated by the tag sending a known preamble in~\cite{BackFiSigcom15}, and estimated by using least-square (LS) algorithm at the receiver in~\cite{YangBackscatter15}.}}~\cite{BackFiSigcom15,YangBackscatter15}.}

Since the backscattered signal is transmitted at the same frequency as the direct-link signal, the CABC system in Fig.~\ref{fig:Fig1} can be considered as a spectrum sharing system~\cite{KangLiangICC17, YCLiangTVT15}. The detection of $\bs(n)$ and $c(n)$ has to consider the mutual effect of the direct-link and the backscatter-link. Specifically, there are two main challenges for signal detection at the C-RX, which are listed as follows: (1) First, since the direct-link channels $\bh_1$ is typically much stronger than the backscatter-link channel $\bh_2$, the signal-to-interference-noise ratio (SINR) for the C-RX to detect the A-BD signal $c(n)$ is very low, if the direct-link signal $\bh_1 s_k(n)$ is treated as interference; (2) While Eq.~\eqref{eq:rx_sig_mtx} looks like a multiple-input-multiple-output (MIMO) model, for each $k$, the two data streams $s_k(n)$ and $c(n)s_k(n)$ are mutually dependent, which makes the receiver design more challenging.


\section{Receiver Design For CABC Under Flat Fading Channels}\label{sec:receiver_flat}
In this section, we design the optimal ML detector, suboptimal linear detectors and SIC-based detectors for the CABC system under flat fading channels.

\vspace{-0.2cm}
\subsection{ML Detector}\label{sec:OptDetection}
The ML estimate $\hatbx_{\sf ml} (n) \triangleq [\hats_0(n), \ \ldots, \ \hats_{K-1}(n), $ $\ \hatc(n)]^T$ is given by~\cite{JGProakisMSalehi05}%
\begin{align} \label{eq:ML_flat}
  \hatbx_{\sf ml} (n)  \!=\!  \underset{ \substack{c(n) \in \calA_{\sf c}, \\ s_k (n) \in \calA_{\sf s}, \forall k}}{\arg \min} \! \sum \limits_{k=0}^{K-1} \! \left\| \by_k (n) \!-\! \bh_1 s_k(n)  \!-\! \bh_2 s_k(n) c(n) \right\|^2.
\end{align}

The number of search in the above ML detector is $|\calA_{\sf c}|  |\calA_{\sf s}|^K$, which grows exponentially as the modulation size $|\calA_{\sf s}|$ increases, resulting into extremely high complexity.

Fortunately, the complexity of ML detection can be reduced significantly by making use of the structure of the received signals. Since the A-BD signal $c(n)$ keeps constant for $s_{k}(n)$, $k=0,\ldots,K-1$, we can obtain the ML estimate of $s_k(n)$ conditioned on each $c(n)$ candidate, denoted as $\hats_k(n) |_{c(n)}$. Based on all the conditional estimates $\hats_k(n) |_{c(n)}$'s, we can obtain the ML estimate of $c(n)$, denoted as $\hatc(n)$. Finally, the conditional estimate $\hats_k(n) = \hats_k(n) |_{\hatc(n)}$ corresponding to $\hatc(n)$ is the ML estimate of $s_k(n)$. The details of the low-complexity ML detector are described as follows.

\emph{1) Estimating $\bs(n)$ for a given $c(n)$}: For a given $c(n)$ candidate, the signal received by the $m$-th antenna in the $k$-th symbol period can be rewritten from~\eqref{eq:yn_1} as
\begin{align}\label{eq:receives yk}
\by_k(n) = \tilbh |_{c(n)} s_k (n) + \bu_k (n),
\end{align}
where the equivalent channel $\tilbh |_{c(n)} \triangleq \bh_1 + \bh_2 c(n)$. By applying maximum-ratio-combining (MRC) to the signal vector $\by_k(n)$, $s_k(n)$ can be estimated as follows 
\begin{align}\label{eq:ML_flat_step1}
  \hats_k(n) |_{c(n)}=   \underset{s_k(n) \in \calA_{\sf s}}{\arg \min} \left| \frac{\tilbh^H |_{c(n)}}{\| \tilbh |_{c(n)} \|^2} \by_k(n) - s_k(n) \right|^2,
\end{align}
for $k=0, \ldots, K-1$.

\emph{2) Estimating the optimal $c(n)$} and $\bs(n)$}: The optimal A-BD signal can be further estimated as follows
\begin{align}\label{eq:ML_flat_step2}
  &\hatc (n) =   \\
  &\underset{c(n) \in \calA_{\sf c}}{\arg \min} \sum \limits_{k=0}^{K-1} \left\| \by_k (n)  \!-\!  \bh_1 \hats_k (n) |_{c(n)} \!-\! \bh_2 c(n) \hats_k (n) |_{c(n)} \right\|^2. \nonumber
\end{align}
Finally, the optimal $s_k(n)$ becomes: $\hats_k(n) = \hats_k(n) |_{\hatc(n)}$, $\forall k=0, \ldots, K-1$.

The number of search in the above two-step ML detector is $K |\calA_{\sf c}| |\calA_{\sf s}|$, which is lower than that of the original ML detector in~\eqref{eq:ML_flat}. \black{Notice that the number of search for the two-step ML detector is still large, for the case of large ratio $K$ and high order modulation at the RF source and/or the A-BD. This motivates us to derive suboptimal detectors with much lower complexity in the next two subsections.}


\vspace{-0.2cm}
\subsection{Linear Detectors}\label{sec:LinearDetection}
For notational convenience, we denote the block-diagonal channel matrix $\tilbH=\diag \{\bH, \ \ldots, \ \bH\} \in \calC^{MK \times 2K}$, the transmit signal vector $\tilbx (n) =[\bx_0 (n) ; \ \ldots; \ \bx_{K-1} (n) ] \in \calC^{2K \times 1}$, the noise vector $\tilbu (n) =[\bu_0 (n) ; \ \ldots; \ \bu_{K-1} (n) ]^T \in \calC^{MK \times 1}$, and the received signal vector $\tilby  (n) =[\by_0 (n) ; \ \ldots; \ \by_{K-1} (n) ]\in \calC^{MK \times 1}$. Thus, the received signals can be rewritten as follows
\begin{align}\label{eq:rx_sig_mtx_composite}
  \tilby(n) = \tilbH \tilbx (n) + \tilbu(n).
\end{align}

For linear detectors, the C-RX applies a block-diagonal decoding matrix $\tilbT = \diag \{\bT_0; \ \ldots, \ $ $\bT_{K-1}\} \in \calC^{2K \times MK}$ with each matrix $\bT_k \in \calC^{2 \times M}$, to extract the signals from both the RF source and the A-BD, i.e.,
\begin{align}
  {\overline{\bx}}(n) = \tilbT \tilby (n). \label{eq:linear_receiver}
\end{align}
For MRC, zero-forcing (ZF) and minimum mean-square-error (MMSE) detectors, each matrix $\bT_k$ in the decoding matrix $\tilbT$, for $k=0, \ldots, K-1$, is given as follows~\cite{GoldsmithWC2005}, respectively,
\begin{align}
\bT_k  \!=\! \left\{ \!
\begin{array}{cl}
\left[ \frac{\bh_1^H }{\left\| \bh_1 \right\|^2}; \frac{\bh_2^H }{\left\| \bh_2 \right\|^2} \right], &\text{for \ \ MRC}  \\
    \left( \bH^H \bH  \right)^{-1} \bH^H , &\text{for \ \ ZF}  \\
    \left( \bH^H  \bH   + \frac{\sigma^2}{P_{\sf s}} \bI_2 \right)^{-1} \bH^H , &\text{for \ \ MMSE}.
  \end{array}
  \right.
\end{align}

After the linear detection, the RF-source and A-BD symbols are recovered as follows
\begin{align}
  \hats_k(n)& \!=\! \underset{s_k(n) \in \calA_{\sf s}}{\arg \min} \left | s_k(n)  \!-\! \barx_{2k+1}(n)\right |, \forall k=0, \ldots, K \!-\! 1 \label{eq:sn_est}\\
  \hatc(n)&=\underset{c(n) \in \calA_{\sf c}}{\arg \min} \sum \limits_{k=0}^{K-1} \left | c(n) - \frac{\barx_{2k+2}(n)}{\hats_k(n)}\right |.
\end{align}

\subsection{SIC-based Detectors}\label{sec:NonlinearDetection}
Since the backscatter-link channel suffers from double fading, the direct-link channel $\bh_1$ is typically stronger than the backscatter-link channel $\bh_2$. As a result, from \eqref{eq:yn_1}, the C-RX first obtains the estimate of the RF-source signals $s_k(n)$ using \eqref{eq:linear_receiver} and \eqref{eq:sn_est}; then subtracts the direct-link interference from the RF source, and detects the A-BD signal $c(n)$; finally, it obtains a refined estimate of $s_k(n)$ by exploiting the recovered A-BD signal. The details of the second and third steps of the SIC-based detector are described as follows.

\subsubsection{Second Step for Estimating $c(n)$}
After obtaining $\hats_k(n)$ from \eqref{eq:sn_est}, we subtract the direct-link interference $\bh_1\hats_k(n)$ from the received signal $\by_k(n)$, yielding the following intermediate signal
\begin{align}
  \bv_k (n) = \by_k(n) - \bh_1 \hats_k(n).
\end{align}
Then, the C-RX applies the MRC detector $\bt_2 = \frac{\bh_2^H}{\| \bh_2 \|^2}$ to the intermediate signal $\bv_k (n)$, and obtains
\begin{align}
  \tily_{2,k}(n) =\bt_2 \bv_k (n).
\end{align}
The A-BD signal $c(n)$ is finally recovered as follows
\begin{align}\label{eq:d_est}
  \hatc(n)=\underset{c(n) \in \calA_{\sf c}}{\arg \min} \sum \limits_{k=0}^{K-1} \left | c(n) - \frac{\tily_{2, k}(n)}{\hats_k(n)} \right |.
\end{align}

\subsubsection{Third Step for Re-Estimating $s(n)$}

From \eqref{eq:yn_1}, the received signal can be rewritten as
\begin{align}
  \by_k (n) = \tilbh |_{c(n)} s_k(n) + \bu_k(n), \label{eq:yn_2}
\end{align}
where $\tilbh |_{c(n)} = \bh_1 + \bh_2 c(n)$. Once we have $\hatc(n)$, we can construct $\widehat{\bw}(n) = \bh_1  + \bh_2 \hatc(n)$, and re-estimate $s_k(n)$ as follows
\begin{align} \label{eq:s_est}
  \hats_k^{\star}(n)=\underset{s_k(n) \in \calA_{\sf s}}{\arg \min} \left | s_k(n) - \frac{\widehat{\bw}^H (n)}{\| \widehat{\bw} (n)\|^2} \by_k(n) \right |.
\end{align}
When MRC, ZF and MMSE estimator are used in the first step for estimating $s_k(n)$, the detector is referred to MRC-SIC detector, ZF-SIC detector and MMSE-SIC detector, respectively.


\section{CABC Under Frequency-Selective Fading Channels}\label{sec:OFDM}

In this section, we study the CABC system under frequency-selective fading channels. OFDM signals are considered as the ambient RF signals as OFDM has been widely adopted in wireless standards, such as WiFi, DVB, and LTE~\cite{GoldsmithWC2005}.


\subsection{Signal Model}\label{subsec:signal_ofdm}
Let $N$ be the number of subcarriers of the OFDM modulation, and ${\mathbf s} (n) = [s_0(n),\ldots, s_{N-1}(n)]^T \in \calC^{N}$ the $n$-th OFDM symbol of the RF source. After inverse discrete Fourier transform (IDFT) operation at the RF source, a cyclic-prefix (CP) of length $N_{\sf c}$ is added at the beginning of each OFDM symbol. We design the symbol period of the A-BD signal $c(n)$ to be the same as the OFDM symbol period. \black{We assume that the A-BD can align the transmission of its own symbol $c(n)$ with its received OFDM symbol\footnote{The effect of imperfect timing synchronization at the A-BD is simulated in Section~\ref{sec:simulations_ofdm_2}. The detection of $c(n)$ is shown to be robust to the imperfect A-BD synchronization, due to the spreading gain and diversity gain with frequency-selective fading channels.}, since the A-BD can estimate the arrival time of OFDM signal by some methods like the scheme that utilizes the repeating structure of CP in OFDM singals~\cite{YangLiangZhangPei17}.} 

The system model is similar to Fig.~\ref{fig:Fig1}. We consider the block fading channel model, where the channel coefficient remains the same within each block but may change among blocks. We assume that the channel block length is much longer than the OFDM symbol period. Let $f_{m, l}$ be the $L_{\sf f}$-path channel response between the RF source and the $m$-th receive antenna, for $m=1,\ \ldots, \ M$, at the C-RX, $v_{l}$ be the $L_{\sf v}$-path channel response between the RF source and the A-BD, $g_{m, l}$ be the $L_{\sf g}$-path channel response between the A-BD and the $m$-th receive antenna at the C-RX. For the direct-link channel, define the frequency response of the $k$-th subcarrier as $\lambda_{m,k} = \sum_{l=0}^{L_{\sf f}-1} f_{m,l}e^{\frac{-j 2\pi k l}{N}}$ for $k = 0,\ldots ,N-1$. Similarly, for the backscatter-link channel, define the frequency response of the $k$-th subcarrier as $\delta_{m,k} = \alpha \sum_{l=0}^{L_{\sf v}-1} v_{l}e^{\frac{- j 2 \pi k l }{N}} \times \sum_{l=0}^{L_{\sf g}-1} g_{m,l}e^{\frac{- j 2 \pi k l }{N}}$, which contains the effect of the double channel fading.
The transmitted time-domain signal in each symbol period is given by
\begin{align}
  x_q (n) \!=\! \sum_{k=0}^{N-1} s_k(n) e^{j 2\pi \frac{qk}{N}}, \text{for} \;q \!=\! 0, \ 1, \ \ldots, \ N \!-\! 1.
\end{align}

For convenience, we assume that the C-RX is timing synchronized to the arrival time of the direct-link signal. The arrival of the backscatter-link signal at the C-RX is typically delayed by a small time $d$ ($d \geq 0$), compared to the arrival of the direct-link signal. The signal received at the $m$-th antenna of the C-RX can be written as
\begin{align}\label{eq:receives signal0}
  y_{m, q} &\!=\!\alpha \sqrt{P_{\sf s}} c(n) \sum_{l_2=0}^{L_{\sf g}-1} \sum_{l_1=0}^{L_{\sf v}-1} x_{q-l_1-l_2-d}(n) v_{l_1} g_{m,l_2}\!+\!... \nonumber \\
  &\qquad \sqrt{P_{\sf s}} \sum_{l=0}^{L_{\sf f}-1} x_{q-l}(n) f_{m,l}+u_{m,q}(n),
\end{align}
where the noise $u_{m,q}(n) \sim \calC \calN(0, \sigma^2)$.


We assume that the delay $d$ is sufficiently small\footnote{For extreme cases with delay $d>d_{\max}$, there exists interblock-interference (IBI) and inter-channel-interference (ICI) for detecting $\bs(n)$, and ISI for detecting $c(n)$. The effect of IBI and ICI will be numerically shown in Section~\ref{sec:simulations_ofdm}.} such that $d \leq d_{\max} \triangleq N_{\sf c}-L_{\sf v}-L_{\sf g}+2$. After removing the CP, the C-RX takes the time window $[N_{\sf c}, \ N_{\sf c}+N-1]$ for discrete-Fourier-transform (DFT) operation. From the (circular) time-shift property of the Fourier transform, the output signal at the $k$-th subcarrier by the $m$-th antenna can be written from~\eqref{eq:receives signal0} as~\cite{MorelliKuoProceed07}
\begin{align}
{z_{m, k}}& (n) = \sum_{q=0}^{N-1} y_{m,q} (n)  e^{-j 2\pi \frac{qk}{N}} \label{eq:receives signal3} \\
&\eqa \lambda_{m,k} s_k (n) \!+\! c(n) \delta_{m,k} s_k (n) e^{-j 2\pi \frac{dk}{N}} \!+\! \tilu_{m, k} (n), \nonumber
\end{align}
where the frequency-domain noise $\tilu_{m, k} (n) \sim \calC \calN(0, \sigma^2)$.

From~\eqref{eq:receives signal3}, it is observed that the signal backscattered by the A-BD is the multiplication of a low-rate A-BD signal $c(n)$ and the high-rate spreading code signal $s_k(n)$ in an over-the-air manner. The corresponding spreading gain for transmitting the A-BD signal $c(n)$ is $N$.

Define the signals received by all antennas at the $k$-th subcarrier as the vector $\tilbz_k(n)=[z_{1, k}(n), \  \ldots, \ $ $ z_{M,k}(n)]^T$, which can be rewritten from~\eqref{eq:receives signal3} as follows
\begin{align}\label{eq:zk_1}
  \tilbz_k(n) = \bh_{{\sf d}, k} s_k(n) + \bh_{{\sf b}, k} s_k(n) c(n) +\tilbu_k(n),
\end{align}
where the direct-link channel vector $\bh_{{\sf d}, k} = [\lambda_{1,k},  \ldots,  \lambda_{M,k}]^T$, the backscatter-link channel vector $\bh_{{\sf b}, k} = e^{-j 2\pi \frac{dk}{N}} \cdot [\delta_{1,k},  \ldots,  \delta_{M,k}]^T$, and the noise vector $\tilbu_k(n)=[\tilu_{1, k}(n), \ldots, \tilu_{M, k}(n)]^T$.


\subsection{Optimal ML Detector}
Notice that the signal model~\eqref{eq:zk_1} has the same structure as the signal model~\eqref{eq:yn_1} under flat-fading channels in Section~\ref{sec:model_flat}. Hence, we directly present the low-complexity ML detector.

\emph{1) Estimating $\bs(n)$ for given $c(n)$}: For a given $c(n)$ candidate, the signal received by the $m$-th antenna at the $k$-th subcarrier can be rewritten as
\begin{eqnarray}\label{eq:receives signal4}
{z_{m, k}(n)} = \tilH_{m, k} |_{c(n)} s_k (n) + \tilu_{m, k} (n),
\end{eqnarray}
where the equivalent channel $\tilH_{m, k} |_{c(n)} \triangleq \lambda_{m,k} + c(n) \delta_{m,k} e^{-j 2\pi \frac{dk}{N}}$. Define the equivalent channel vector $\tilbh_k |_{c(n)} =[\tilH_{1, k} |_{c(n)},  \  \ldots, \ \tilH_{M, k} |_{c(n)}]^T$. Given $c(n)$, by applying MRC to the signal vector $\tilbz_k(n)$, the $s_k(n)$ can be estimated and quantized as follows 
\begin{align}\label{eq:ML_ofdm_step1}
  \hats_k(n) |_{c(n)} \!=\!   \underset{s_k(n) \in \calA_{\sf s}}{\arg \min} \left| \frac{\tilbh_k^H |_{c(n)}}{\| \tilbh_k |_{c(n)} \|^2} \tilbz_k(n) \!-\! s_k(n) \right|^2,
\end{align}
for $k=0, \ldots, N-1$.


\emph{2) Estimating the optimal $c(n)$ and $\bs(n)$}: The optimal A-BD signal can be estimated as follows
\begin{align}\label{eq:ML_ofdm_step2}
  \hatc (n)  \!=\!   \underset{c(n) \in \calA_{\sf c}}{\arg \min} \! \sum \limits_{m=1}^M \! \sum \limits_{k=0}^{N-1} \left| z_{m, k} (n) \!-\! \tilH_{m, k} |_{c(n)} \hats_k (n) |_{c(n)} \right|^2.
\end{align}
Finally, the optimal $s_k(n)$ is $\hats_k(n) = \hats_k(n) |_{\hatc(n)}, \; \forall k=0, \ldots, N-1$.

Different from the flat fading channel case in Section~\ref{sec:OptDetection}, the estimation of $c(n)$ in~\eqref{eq:ML_ofdm_step2} benefits not only from the spreading gain, but also from the frequency diversity.
\section{Performance Analysis}\label{sec:ber_analysis}
In this section, we analyze the data rate and error rate performance for the proposed CABC system. For analytical convenience, we assume that the RF source and the A-BD adopt the quadrature phase shift keying (QPSK) and the binary phase shift keying (BPSK) modulation, respectively. That is, $\calA_{\sf s} = (\pm 1 \pm j)/\sqrt{2}$, and $\calA_{\sf c} = \pm 1$. However, the analytical method can be generalized to other modulation schemes~\cite{JGProakisMSalehi05}.

\subsection{Flat Fading Channels}
\subsubsection{A-BD Data Rate Performance}
{{\textcolor{black}{The data (symbol) rate of the A-BD is given by}}}
\begin{align}\label{eq:rate}
  R_{\sf c} = \frac{R_{\sf s}}{K}.
\end{align}

\subsubsection{BER Performance for ML detector}\label{sec:ber_analysis_ML}
In this subsection, we analyze the BER performance for the case of $K=1$, for ease of exposition. We ignore the subscript $k=0$ and use the notation $s(n)$, for simplicity. However, the analytical method can be generalized to other cases of $K>1$. In fact, for cases of $K>1$, the BER analysis is analogous to that for ML detector for the CABC system over ambient OFDM carriers in Section~\ref{sec:analyze_ofdm_ML}, thus omitted herein.

Given $\bH$, we denote the BERs of $s(n)$ and $c(n)$ by $P_{\sf e, s} (\bH)$ and $P_{\sf e, c} (\bH)$, respectively. For ML detector, we have the following theorem on the BER.
\begin{mythe}\label{theorem:BERML}
Given $\bH$, the BERs of using ML detector to detect $s(n)$ and $c(n)$ are given as follows, respectively,
\begin{align}
  P_{\sf e, s} (\bH) &= \frac{-C_1(\bH)-\sqrt{C_1^2(\bH) \!- \!4C_2(\bH)C_0(\bH)}}{2C_2(\bH)}, \label{eq:BER_s_analysis} \\
  P_{\sf e, c} (\bH) &= \frac{P_{\sf e, s} (\bH)-a_1(\bH)}{a_2(\bH) - a_1(\bH)},\label{eq:BER_t_analysis}
\end{align}
where the coefficients
\begin{align}
C_0(\bH)&= b_1(\bH) a_2(\bH) +a_1(\bH) \left[1-b_1(\bH)\right], \\
C_1(\bH)&= \left[1-2b_1(\bH)\right] \left[a_2(\bH)-a_1(\bH)\right]-1,\\
C_2(\bH)&= \left[b_1(\bH)+b_2(\bH)-1\right]\left[a_2(\bH)-a_1(\bH)\right],
\end{align}
where the coefficients each of which represents the BER of $s(n)$ or $c(n)$ for certain condition, are given by
\begin{align}
  a_1 (\bH) \!&= \! \frac{1}{2} Q \left( \! \! \frac{ \| \bh_1 \!+\! \bh_2 \|}{\sigma^2} \! \! \right) \!+\!  \frac{1}{2} Q \left( \! \! \frac{ \| \bh_1 \!-\! \bh_2 \|}{\sigma} \! \! \right),  \label{eq:a1}\\
    a_2 (\bH) \!&= \frac{1}{4}\calQ \left( \frac{\| \bh_1 -\bh_2\|(\theta_{R,1}(\bH) + \theta_{I,1}(\bH))}{\sigma}\right) +...\nonumber \\
    &\quad \frac{1}{4}\calQ \left( \frac{\| \bh_1 -\bh_2\|(\theta_{R,1}(\bH) - \theta_{I,1}(\bH))}{\sigma}\right) +... \nonumber \\
   &\quad \frac{1}{4}\calQ \left( \frac{\| \bh_1 +\bh_2\|(\theta_{R,2}(\bH) + \theta_{I,2}(\bH))}{\sigma}\right) +...\nonumber \\
   &\quad \frac{1}{4}\calQ \left( \frac{\| \bh_1 +\bh_2\|(\theta_{R,2}(\bH) - \theta_{I,2}(\bH))}{\sigma}\right),\label{eq:a2}\\
  b_1 (\bH) &= Q \left( \frac{\sqrt{2} \| \bh_2 \|}{\sigma} \right),  \label{eq:b1} \\
  b_2(\bH) &=\frac{1}{2} \Big[\! \calQ \! \left( \! \frac{\sqrt{2}\| \bh_2\|}{\sigma} \! \left(\!  -1 \!+\! 2 \textrm{\em Re} \left\{ \frac{\bh_2^H \bh_1}{\|\bh_2\|^2} \right\} \! \right) \! \right) +...\nonumber \\
  &\quad \calQ \left( \frac{\sqrt{2}\| \bh_2\|}{\sigma} \left( -1-2 \textrm{\em Re} \left\{ \frac{\bh_2^H \bh_1}{\|\bh_2\|^2}\right\} \right)\right)\Big],\label{eq:b2}
\end{align}
{{\textcolor{black}{where the $\calQ$-function $\calQ(z) \!=\! \frac{1}{\sqrt{2 \pi}} \int_z^{\infty} e^{-u^2/2} du$, the expressions $\theta_{R,1} (\bH) \!=\! \textrm{\em Re} \left\{ \frac{(\bh_1-\bh_2)^H (\bh_1+\bh_2)}{\| \bh_1-\bh_2\|^2} \right\}, \theta_{I,1} (\bH) =$ $\textrm{\em Im} \left\{ \! \frac{(\bh_1 \!-\!\bh_2)^H (\bh_1 \!+\! \bh_2)}{\| \bh_1 \!-\! \bh_2\|^2} \! \right\}$, $\theta_{R,2}(\bH) \!=\! \textrm{\em Re} \left\{ \frac{(\bh_1+\bh_2)^H (\bh_1-\bh_2)}{\| \bh_1+\bh_2\|^2} \right\}$ and $\theta_{I,2}(\bH) \!=\! \textrm{\em Im} \left( \frac{(\bh_1+\bh_2)^H (\bh_1-\bh_2)}{\| \bh_1+\bh_2\|^2} \right)$.}}}
\end{mythe}

\begin{proof}
See proofs in Appendix~\ref{proof_Thm_ML}.
\end{proof}

By taking the expectation over the channel $\bH$, the average BERs are obtained as $\bar{P}_{\sf e, s} = \bbE_{\bH} \left[P_{\sf e, s} (\bH)\right]$ and $\bar{P}_{\sf e, c} = \bbE_{\bH} \left[P_{\sf e, c} (\bH)\right]$, respectively.

\subsubsection{BER Performance for Linear Detectors}\label{eq:ber_linear}
For MRC detector, we have the following proposition on the BER.
\begin{mypro}\label{theorem:BERMRC}
Given $\bH$, the BERs of using MRC detector to detect $s(n)$ and $c(n)$ are {{\textcolor{black}{given in~\eqref{eq:BER_s_analysis_mrc} and~\eqref{eq:BER_t_analysis_mrc} at the top of the next page, respectively.}}}
\begin{figure*}[!t]
\begin{align}
P_{\sf e, s} (\bH) &= \frac{1}{4}\calQ \left( \frac{\| \bh_1 \|}{\sigma} \left( 1+ \frac{\textrm{Re} \left\{\bh_1^H\bh_2\right\}-\textrm{Im} \left\{\bh_1^H\bh_2\right\}}{\| \bh_1\|^2}\right)\right) \!+\! \frac{1}{4}\calQ \left( \frac{\| \bh_1 \|}{\sigma} \left( 1+ \frac{\textrm{Re} \left\{\bh_1^H\bh_2\right\}+\textrm{Im} \left\{\bh_1^H\bh_2\right\}}{\| \bh_1\|^2}\right)\right)+... \nonumber \\
  &\quad \;\; \frac{1}{4}\calQ \left( \! \frac{\| \bh_1 \|}{\sigma} \left( \! 1 \!-\! \frac{\textrm{Re} \left\{\bh_1^H\bh_2\right\}-\textrm{Im} \left\{\bh_1^H\bh_2\right\}}{\| \bh_1\|^2} \! \right) \! \right) \!+\! \frac{1}{4}\calQ \left( \! \frac{\| \bh_1 \|}{\sigma} \left( \! 1 \!-\! \frac{\textrm{Re} \left\{\bh_1^H\bh_2\right\}+\textrm{Im} \left\{\bh_1^H\bh_2\right\}}{\| \bh_1\|^2}\! \right) \!\right), \label{eq:BER_s_analysis_mrc} \\
P_{\sf e, c} (\bH) &=\frac{\left(1\!-\!P_{\sf e, s} (\bH)\right)^2}{2}\calQ \left(\! \frac{\sqrt{2} \| \bh_2 \|}{\sigma} \left(\! 1 \!+\! \frac{\textrm{Re} \left\{\bh_2^H\bh_1\right\}}{\| \bh_2\|^2}\! \right) \!\right)+\frac{\left(1 \!-\! P_{\sf e, s} (\bH)\right)^2}{2}\calQ \left(\! \frac{\sqrt{2} \| \bh_2 \|}{\sigma} \left( \! 1\!- \!\frac{\textrm{Re} \left\{\bh_2^H\bh_1\right\}}{\| \bh_2\|^2}\! \right) \! \right) +... \nonumber \\
  &\quad \;\; \frac{P_{\sf e, s}^2 (\bH)}{2}\calQ \left( \frac{\sqrt{2} \| \bh_2 \|}{\sigma} \left( -1- \frac{\textrm{Re} \left\{\bh_2^H\bh_1\right\}}{\| \bh_2\|^2}\right)\right)+\frac{P_{\sf e, s}^2 (\bH)}{2} \calQ \left( \frac{\sqrt{2} \| \bh_2 \|}{\sigma} \left( -1+ \frac{\textrm{Re} \left\{\bh_2^H\bh_1\right\}}{\| \bh_2\|^2}\right)\right)+... \nonumber \\
  &\quad \;\; P_{\sf e, s} (\bH)\left(1-P_{\sf e, s} (\bH)\right). \label{eq:BER_t_analysis_mrc}
\end{align}
\end{figure*}
\end{mypro}

\begin{proof}
{{\textcolor{black}{See proofs in Appendix~\ref{proof_Thm_MRC}.}}}
\end{proof}

Denote the singular vector decomposition (SVD) of $\bH$ as $\bH=\bU \bLambda \bV^H$. Denote the matrix $\bA=(\bH^H \bH)^{-1}=\bV \bLambda^{-2} \bV^H$, with element $A_{ij}$ in its $i$-th row and $j$-th column. For ZF detector, we have the following proposition on the BER.
\begin{mypro}\label{theorem:BERZF}
Given $\bH$, the BERs of using ZF detector to detect $s(n)$ and $c(n)$ are given as follows
\begin{align}
  P_{\sf e, s} (\bH) &= \calQ \left( \frac{1}{\sigma \sqrt{A_{11}(\bH)}}\right), \label{eq:BER_s_analysis_zf} \\
  P_{\sf e, c} (\bH) &= \left(1-P_{\sf e, s} (\bH)\right)^2 \calQ \left( \frac{\sqrt{2}}{\sigma \sqrt{A_{22}(\bH)}}\right)+... \label{eq:BER_t_analysis_zf} \\
  P_{\sf e, s} (&\bH)\left(1 \!-\! P_{\sf e, s} (\bH)\right) \!+\! P_{\sf e, s}^2 (\bH) \calQ \left( \!-\! \frac{\sqrt{2}}{\sigma \sqrt{A_{22}(\bH)}}\right). \nonumber
\end{align}
\end{mypro}

\begin{proof}
{{\textcolor{black}{See proofs in Appendix~\ref{proof_Thm_ZF}.}}}
\end{proof}

For MMSE detector, we have the following proposition on the BER.
\begin{mypro}\label{theorem:BERMMSE}
Given $\bH$, the BERs of using MMSE detector to detect $s(n)$ and $c(n)$ are given as follows
\begin{align}
  &P_{\sf e, s} (\bH) \!=\! \calQ \left( \sqrt{\bh_1^H \left( \bh_2 \bh_2^H + \sigma^2 \bI \right)^{-1} \bh_1}\right), \label{eq:BER_s_analysis_mmse} \\
  &P_{\sf e, c} (\bH) \!=\! P_{\sf e, s}^2 (\bH) \calQ \left( \!-\! \sqrt{\bh_2^H \left( \bh_1 \bh_1^H \!+\! \sigma^2 \bI \right)^{-1} \bh_2}\right)\!+\!...\nonumber \\
  &\qquad \left(\! 1 \!-\! P_{\sf e, s} (\bH) \! \right)^2 \! \calQ \! \left( \! \sqrt{\bh_2^H \left( \bh_1 \bh_1^H \!+\! \sigma^2 \bI \right)^{-1} \bh_2}\right) \!+\!...  \nonumber \\
  &\qquad  P_{\sf e, s} (\bH) \left(1 \!-\! P_{\sf e, s} (\bH)\right). \label{eq:BER_t_analysis_mmse}
\end{align}
\end{mypro}

\begin{proof}
{{\textcolor{black}{See proofs in Appendix~\ref{proof_Thm_MMSE}.}}}
\end{proof}

\subsubsection{BER Performance for SIC-based Detectors}\label{eq:ber_Sic}
For SIC-based detectors, given $\bH$, we denote the BER of $s(n)$ in the first step by $\tilP_{\sf e, s}(\bH)$,
the BER for detecting $c(n)$ in the second step by $\tilP_{\sf e, c} (\bH)$, and the BER of re-estimating $s(n)$ in the third step by $\tilP_{\sf e, s}^{\star} (\bH)$, respectively. We then have the following theorem.
\begin{mythe}\label{theorem:BERSIC}
Given $\bH$, the BERs of using SIC-based detectors to detect $s(n)$ and $c(n)$ are given as follows
\begin{align}
  \tilP_{\sf e, c} (\bH) &= \left(1- \tilP_{\sf e, s} (\bH)\right)^2 b_1 (\bH) +  \label{eq:BER_s_analysis_sic} \\
  &\quad \tilP_{\sf e, s} (\bH) \left(1- \tilP_{\sf e, s} (\bH)\right) + \tilP_{\sf e, s}^2 (\bH)  b_2 (\bH), \nonumber \\
  \tilP_{\sf e, s}^{\star} (\bH) &\!=\! \left(1-\tilP_{\sf e, c} (\bH) \right) a_1 (\bH) + \tilP_{\sf e, c} (\bH) a_2 (\bH), \label{eq:BER_t_analysis_sic}
\end{align}
where $\tilP_{\sf e, s}(\bH)$ is given in~\eqref{eq:BER_s_analysis_mrc}, \eqref{eq:BER_s_analysis_zf} and \eqref{eq:BER_s_analysis_mmse}, for MRC detector, ZF detector and MMSE detector, respectively, $b_1 (\bH)$ and $b_2 (\bH)$ are given in~\eqref{eq:b1} and~\eqref{eq:b2}, and $a_1 (\bH)$ and $a_2 (\bH)$ are given in~\eqref{eq:a1} and~\eqref{eq:a2}, respectively.
\end{mythe}

\begin{proof}
Since conventional linear detector is used in the first step for detecting $s(n)$, the BER $\tilP_{\sf e, s}(\bH)$ is given in~\eqref{eq:BER_s_analysis_mrc}, \eqref{eq:BER_s_analysis_zf} and \eqref{eq:BER_s_analysis_mmse}, for MRC detector, ZF detector and MMSE detector, respectively. Moreover, by using similar steps as in the proof of Theorem~\ref{theorem:BERML}, the BER of detecting $c(n)$ in the second step can be further derived as in~\eqref{eq:BER_s_analysis_sic}, and the BER of re-estimating $s(n)$ in the third step can also be derived as in~\eqref{eq:BER_t_analysis_sic}. This completes the proof.
\end{proof}

The corresponding average BERs are thus $\barP_{\sf e, c} = \bbE_{\bH} \left[ \tilP_{\sf e, c}(\bH)\right]$ and $\barP_{\sf e, s} = \bbE_{\bH} \left[ \tilP_{\sf e, s}^{\star}(\bH)\right]$.

\subsection{Frequency-Selective Fading Channels}\label{sec:analyze_ofdm}
\subsubsection{A-BD Data Rate Performance}
Recall the A-BD symbol period is designed to equal the OFDM symbol period which consists of $(N+N_{\sf c})$ sampling periods, the A-BD data rate is obtained as
\begin{align}\label{eq:rate_ofdm}
  R_{\sf A-BD} = \frac{f_s}{N+N_{\sf c}}.
\end{align}

\subsubsection{BER Performance for ML Detector}\label{sec:analyze_ofdm_ML}
Due to large spreading gain $N$ of detecting $c(n)$ in the second step in~\eqref{eq:ML_ofdm_step2}, the BER of $c(n)$ is in general small, which will be numerically shown in Section~\ref{sec:simulations_ofdm}. We thus focus on the BER of $\bs(n)$ in this subsection.

Denote the composite channel matrix $\calH = \diag\{\bH_0, \ldots, \bH_{N-1}\}$, where the channel matrix for each subcarrier $k$ is $\bH_k = [\bh_{{\sf d}, k}, \ \bh_{{\sf b}, k}]$. For ML detector, we have the following theorem on the BER.



\begin{mythe}\label{theorem:BERML_ofdm}
Given $\calH$ and the BER of $c(n)$ denoted by $P_{\sf e, c} (\calH)$ , the BER of using ML detector to detect $\bs(n)$, denoted by $P_{\sf e, s} (\calH)$, is given as follows
\begin{align}
  P_{\sf e, s} (\calH) &= (1-P_{\sf e, c} (\calH) ) \tila_1 (\calH) + P_{\sf e, c} (\calH) \tila_2 (\calH), \label{eq:BER_s_analysis_ofdm}
\end{align}
where the coefficients each of which represents the BER of $\bs(n)$ for certain condition, are given by
\begin{align}
\tila_1 (\calH) &\!=\! \frac{1}{2N} \! \sum \limits_{k=0}^{N-1} \left[ \! Q \left( \! \! \frac{ \| \bh_{{\sf d}, k} \!+\! \bh_{{\sf b}, k} \|}{\sigma^2} \! \! \right) \!+\! Q \left( \! \! \frac{ \| \bh_{{\sf d}, k} \!-\! \bh_{{\sf b}, k} \|}{\sigma} \! \! \right) \! \right],  \label{eq:a1_ofdm}
\end{align}
\begin{align}
    &\tila_2 (\calH) = \nonumber \\
    &\frac{1}{4N} \! \sum \limits_{k=0}^{N-1} \! \Bigg[ \! \calQ \left(\! \frac{\| \bh_{{\sf d}, k} \!-\! \bh_{{\sf b}, k}\|(\theta_{R,1}(\bH_k) \!+\! \theta_{I,1}(\bH_k))}{\sigma}\!\right) \!+\!...\nonumber \\
    &\quad \calQ \left( \frac{\| \bh_{{\sf d}, k} -\bh_{{\sf b}, k}\|(\theta_{R,1}(\bH_k) - \theta_{I,1}(\bH_k))}{\sigma}\right) \!+\!... \nonumber \\
   &\quad \calQ \left( \frac{\| \bh_{{\sf d}, k} +\bh_{{\sf b}, k}\|(\theta_{R,2}(\bH_k) + \theta_{I,2}(\bH_k))}{\sigma}\right) \!+\!...\nonumber \\
   &\quad \calQ \left( \frac{\| \bh_{{\sf d}, k} +\bh_{{\sf b}, k}\|(\theta_{R,2}(\bH_k) - \theta_{I,2}(\bH_k))}{\sigma}\right) \Bigg],\label{eq:a2_ofdm}
\end{align}
where $\theta_{R,1} (\bH_k) \!=\! \textrm{\em Re} \left\{ \frac{(\bh_{{\sf d}, k}-\bh_{{\sf b}, k})^H (\bh_{{\sf d}, k}+\bh_{{\sf b}, k})}{\| \bh_{{\sf d}, k}-\bh_{{\sf b}, k}\|^2} \right\}, \theta_{I,1} (\bH_k)$ $=\textrm{\em Im} \left\{ \frac{(\bh_{{\sf d}, k}-\bh_{{\sf b}, k})^H (\bh_{{\sf d}, k}+\bh_{{\sf b}, k})}{\| \bh_{{\sf d}, k}-\bh_{{\sf b}, k}\|^2} \right\}$, $\theta_{R,2}(\bH_k) \!=\! \textrm{\em Re} \left\{ \frac{(\bh_{{\sf d}, k}+\bh_{{\sf b}, k})^H (\bh_{{\sf d}, k}-\bh_{{\sf b}, k})}{\| \bh_{{\sf d}, k}+\bh_{{\sf b}, k}\|^2} \right\}$, and $\theta_{I,2}(\bH_k) \!=\! \textrm{\em Im} \left( \frac{(\bh_{{\sf d}, k}+\bh_{{\sf b}, k})^H (\bh_{{\sf d}, k}-\bh_{{\sf b}, k})}{\| \bh_{{\sf d}, k}+\bh_{{\sf b}, k}\|^2} \right)$.
\end{mythe}

\begin{proof}
The BER for each subcarrier $k$ can be proved by using similar steps as in Appendix~\ref{proof_Thm_ML} for proving Theorem~\ref{theorem:BERML}. The overall BER is obtained by averaging over all the $N$ subcarriers. The detailed proof is omitted herein, due to space limitation.
\end{proof}


By taking the expectation over the channel $\calH$, the corresponding average BERs are obtained as $\bar{P}_{\sf e, s} = \bbE_{\calH} \left[P_{\sf e, s} (\calH)\right]$.

\section{Numerical Results}\label{sec:simulation}
In this section, we provide simulation results to evaluate the performance of the proposed detectors for the CABC system. {{\textcolor{black}{We assume that the channel coefficients $f_m$'s, $v$ and $g_m$'s are mutually independent and Rayleigh fading, each of which follows a complex Gaussian distribution with zero mean and some specific variance $\beta_{\sf f}, \ \beta_{\sf v}$, or $\beta_{\sf g}$. We further assume that the channel coefficient variances $\beta_{\sf f} = 10^{-7}$ and $\beta_{\sf v}=10^{-5}$ for larger RF-source-to-C-RX distance and RF-source-to-A-BD distance. The channel coefficient variance $\beta_{\sf g}$,}}} which relates to the short A-BD-to-C-RX distance, is determined by the relative SNR $\Delta \gamma$. \black{Let the backscatter efficiency $\alpha=0.2+0.3j$}. The RF source adopts QPSK modulation, while the A-BD adopts BPSK modulation. The number of antennas at the C-RX is $M=4$, and \textcolor{black}{$10^8$} channel realizations are simulated for average BER evaluation.


\subsection{Flat Fading Channel}
\subsubsection{BER Comparison for ML Detector}\label{subsec:simula_flat_ML}
In this subsection, we first compare the BER performance of the proposed CABC system by using ML detector to that of conventional direct-link SIMO system, for fixed ratio of symbol period $K=1$.


%

Fig.~\ref{fig:Fig2} plots the BERs of the RF-source signal $\bs(n)$ and the A-BD signal $c(n)$ versus the direct-link SNR $\gamma_{\sf d}$ by using the proposed ML detector, respectively. We observe that the proposed low-complexity ML detector achieves the same performance as the original ML detector in~\eqref{eq:ML_flat} which is termed as joint ML detector in Fig.~\ref{fig:Fig2}, and the analytical BERs coincide with the simulated BERs. Hence, we use the low-complexity ML detector in the rest of this section. From Fig.~\ref{fig:Fig2}(a), we observe that the existence of a backscattering A-BD benefits the detection of $\bs(n)$, compared to the conventional direct-link SIMO system. For instance, at a BER level of $10^{-5}$, the CABC system achieves an SNR gain of around 1 dB, for the case of $\Delta \gamma =-10$ dB. The achieved SNR gain becomes larger as the relative SNR $\Delta \gamma$ increases (i.e., the strength of backscattered signal increases). From Fig.~\ref{fig:Fig2}(b), we observe that the BER performance improves as $\Delta \gamma$ increases. Specifically, at a BER level of $10^{-2}$, the CABC system achieves an SNR gain of around 9 dB for $\Delta \gamma =0$ dB, compared to $\Delta \gamma =-10$ dB.


\begin{figure}[!t]
\centering
\subfigure[BERs of RF-source signal $\bs(n)$] {\includegraphics[height=2.5in,width=3.3in,angle=0]{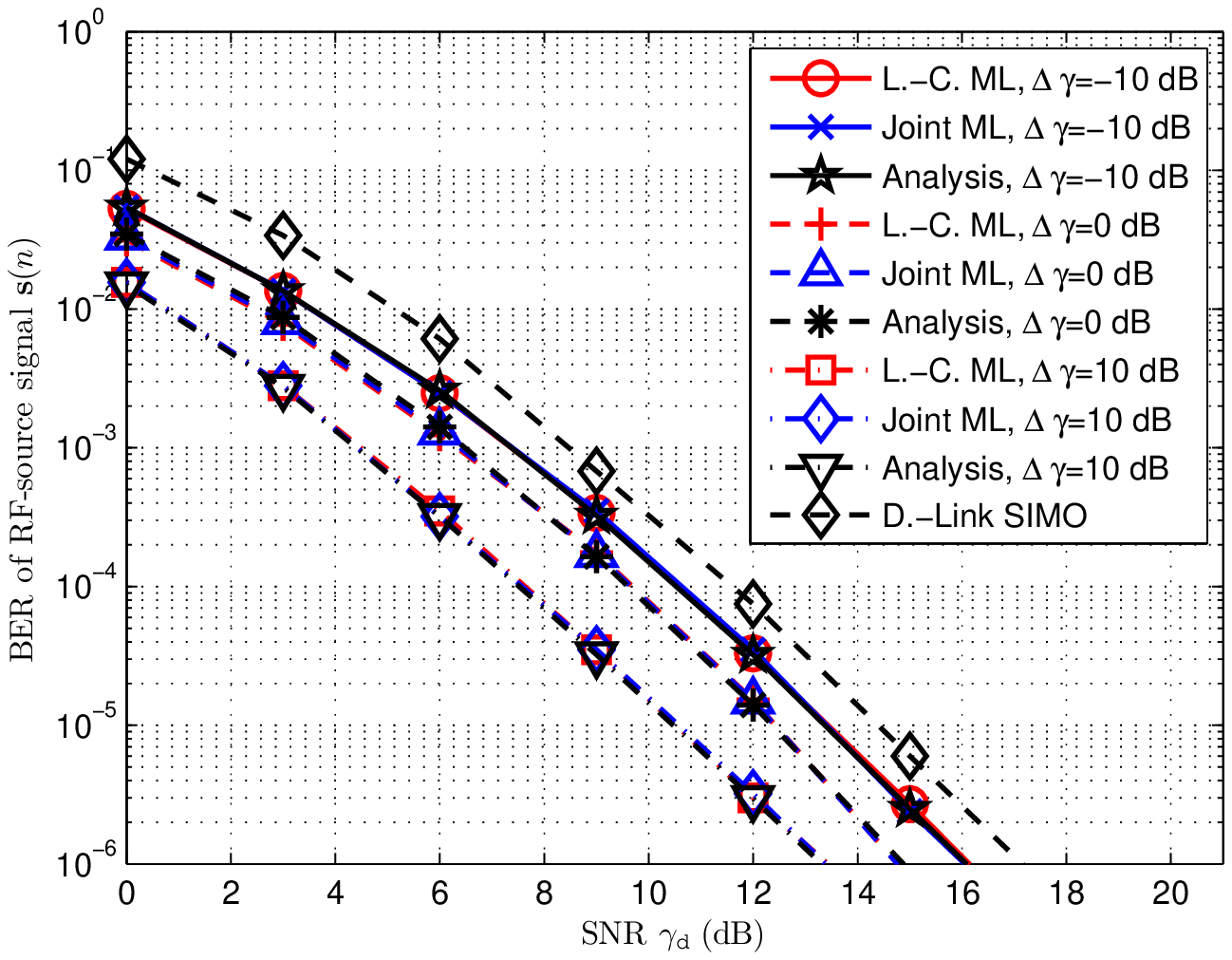}}
\subfigure[BERs of A-BD signal $c(n)$] {\includegraphics[height=2.5in,width=3.3in,angle=0]{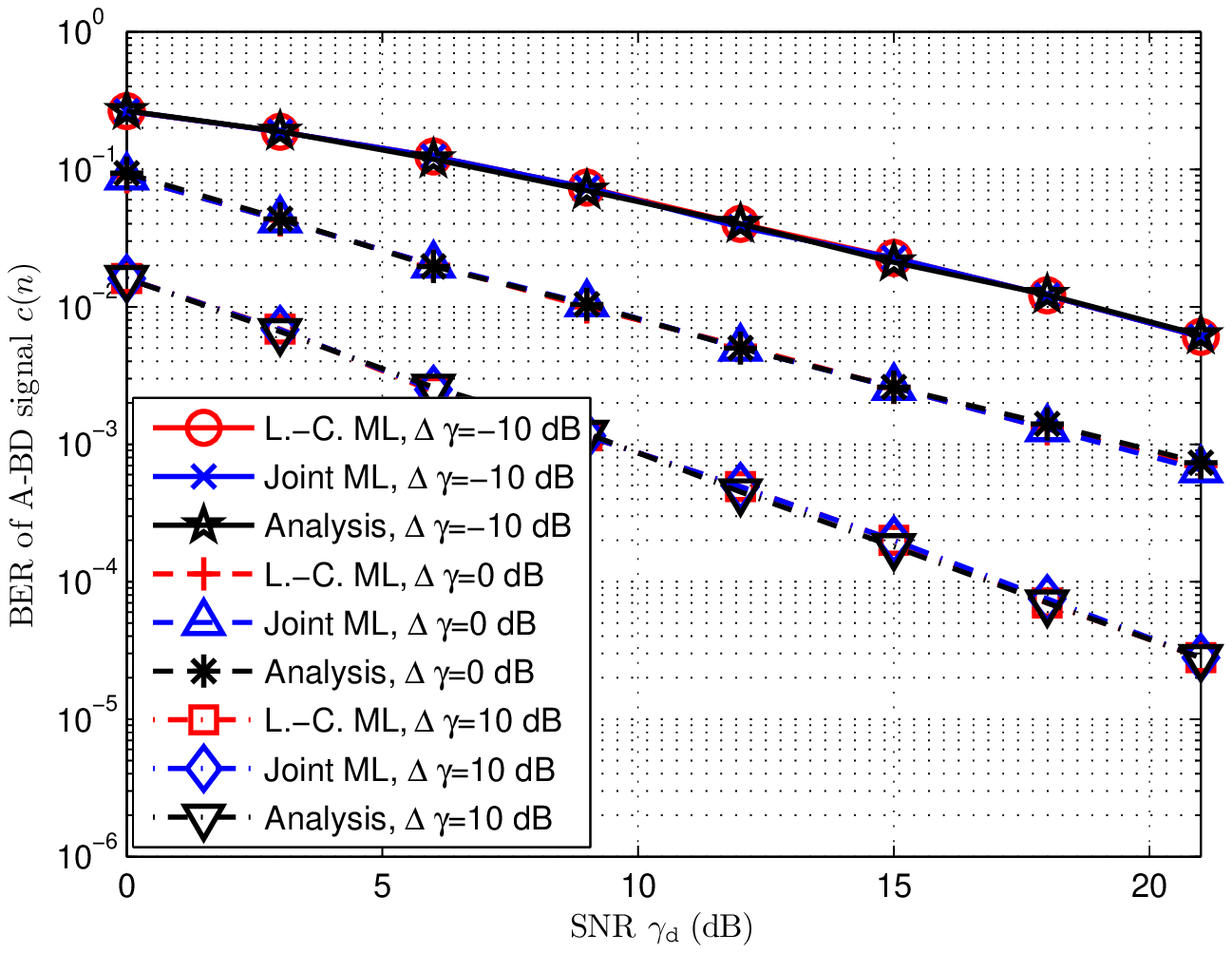}}
\caption{BERs of RF-source signal $\bs(n)$ and A-BD signal $c(n)$ by using ML detector, for $K=1$, {{\textcolor{black}{under flat fading channel}.}}}
\label{fig:Fig2}
\end{figure}

Then, we investigate the impact of the ratio of symbol period $K$ on the BER performance. Fig.~\ref{fig:Fig2Diffk}(a) plots the BERs of $\bs(n)$ versus $\gamma_{\sf d}$ by using ML detector for different $K$'s and different $\Delta \gamma$'s. Given $\Delta \gamma$, the BER performance of $\bs(n)$ improves as $K$ increases. The improvement becomes smaller for larger $K$, which implies that a smaller $K$ is preferable in practice for purpose of higher detection reliability and A-BD data rate.



Fig.~\ref{fig:Fig2Diffk}(b) plots the BERs of $c(n)$ versus the direct-link SNR $\gamma_{\sf d}$ by using ML detector for different $K$'s and different $\Delta \gamma$'s. We observe that for given $\Delta \gamma$, the BER performance of $c(n)$ improves as $K$ increases. In particular, we observe that when $K$ increases by two times, the BER performance of $c(n)$ achieves an SNR gain of around 3 dB, which represents the spreading gain of detecting $c(n)$, but the data rate of the A-BD decreases by half. This observation implies that there is a tradeoff between the data rate and reliability of transmitting $c(n)$ in the proposed CABC system.

\begin{figure}[!t]
\centering
\subfigure[BERs of RF-source signal $\bs(n)$] {\includegraphics[height=2.5in,width=3.3in,angle=0]{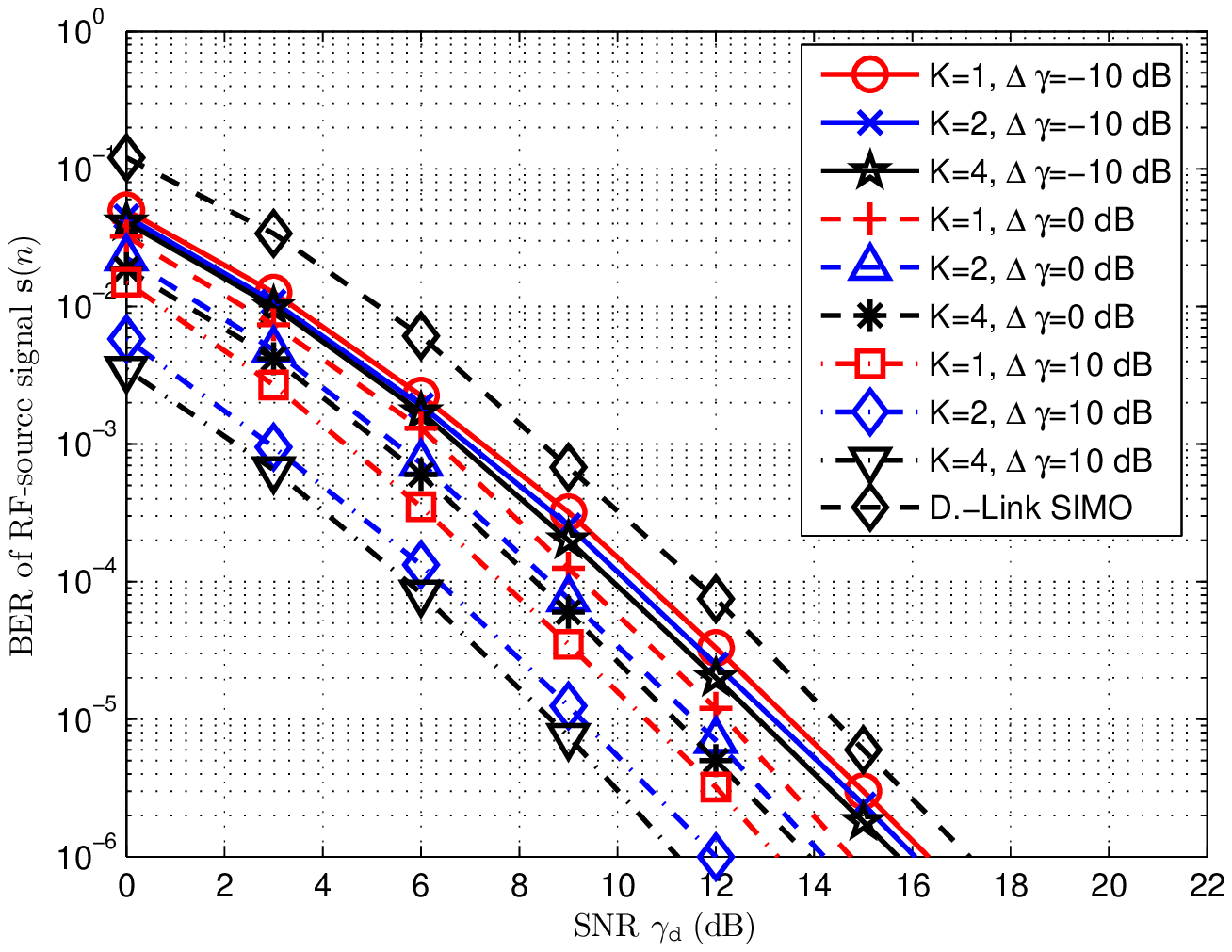}}
\subfigure[BERs of A-BD signal $c(n)$] {\includegraphics[height=2.5in,width=3.3in,angle=0]{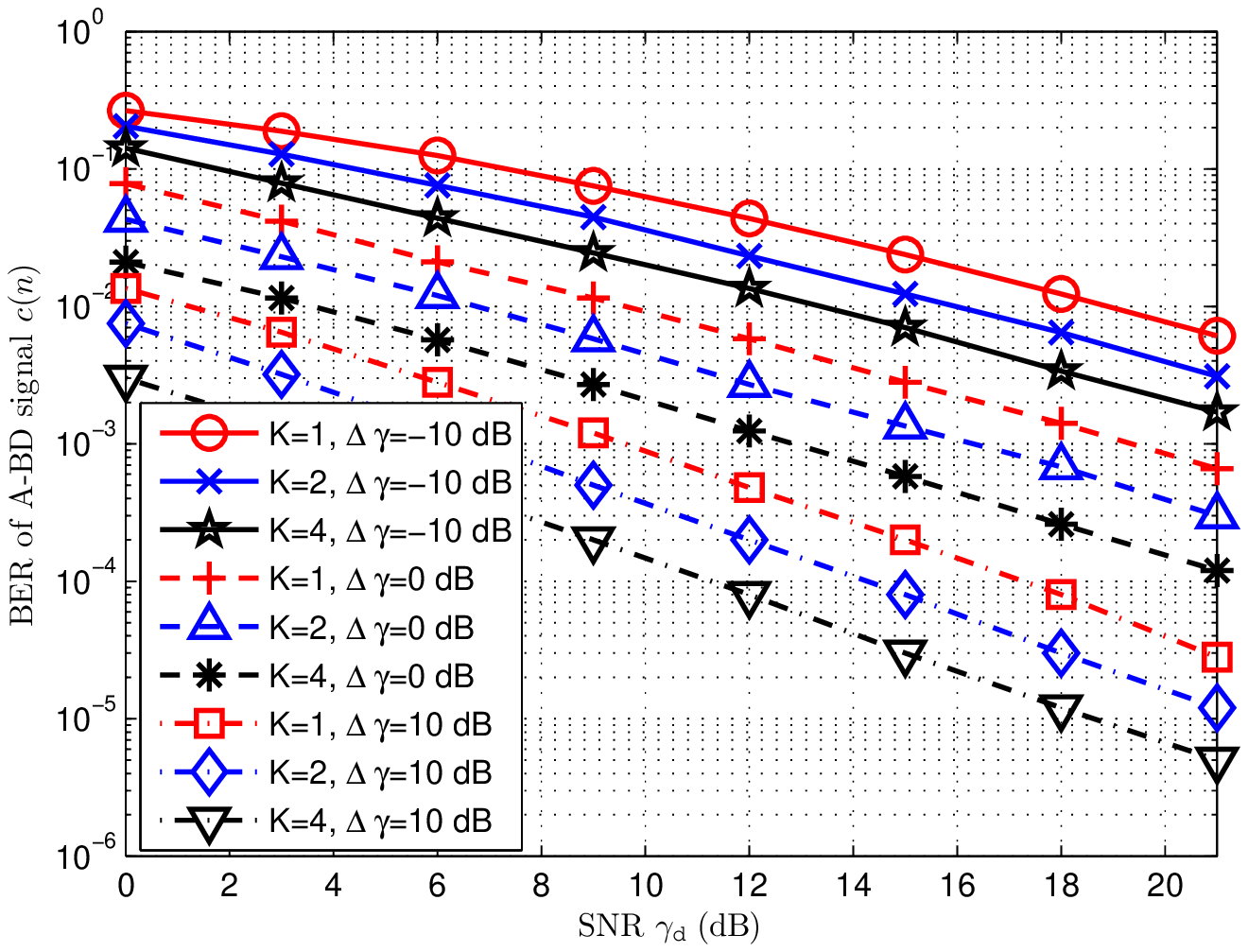}}
\caption{BERs of $\bs(n)$ and $c(n)$ by using ML detector, for different $K$'s, {{\textcolor{black}{under flat fading channel}}}.}
\label{fig:Fig2Diffk}
\end{figure}

\subsubsection{BER Comparison for Suboptimal Detectors}
Let $K=1$. We fix $\Delta \gamma=-10$ dB, i.e., the direct-link signal is 10 times as strong as the backscattered signal, which is typical in practice.

Fig.~\ref{fig:Fig3} compares the BERs of the RF-source signal $\bs(n)$ and the A-BD signal $c(n)$ by using different detectors, respectively. From Fig.~\ref{fig:Fig3}(a), we have the following observations. First, both MRC detector and MRC-SIC detector suffer from BER floors, since the two detectors treat the backscattered signal as interference, and the detection error of A-BD signal $c(n)$ leads to additional interference for detecting $\bs(n)$ using MRC-SIC detector. Second, the BERs for other detectors decrease as the direct-link SNR $\gamma_{\sf d}$ increases. Third, ZF detector suffers from higher BER than MMSE detector. This is explained as follows. Since the direct-link signal is 10 times stronger than the backscattered signal, the optimal detector approximates MRC detector constructed from the direct-link channels. However, ZF detector is constructed to force the interference between the two signals to zero, and thus deviates from the optimal detector. Forth, we observe that the MMSE-SIC detector achieves near-ML detection performance.

\begin{figure}[!t]
\centering
\subfigure[BERs of RF-source signal $\bs(n)$] {\includegraphics[height=2.5in,width=3.3in,angle=0]{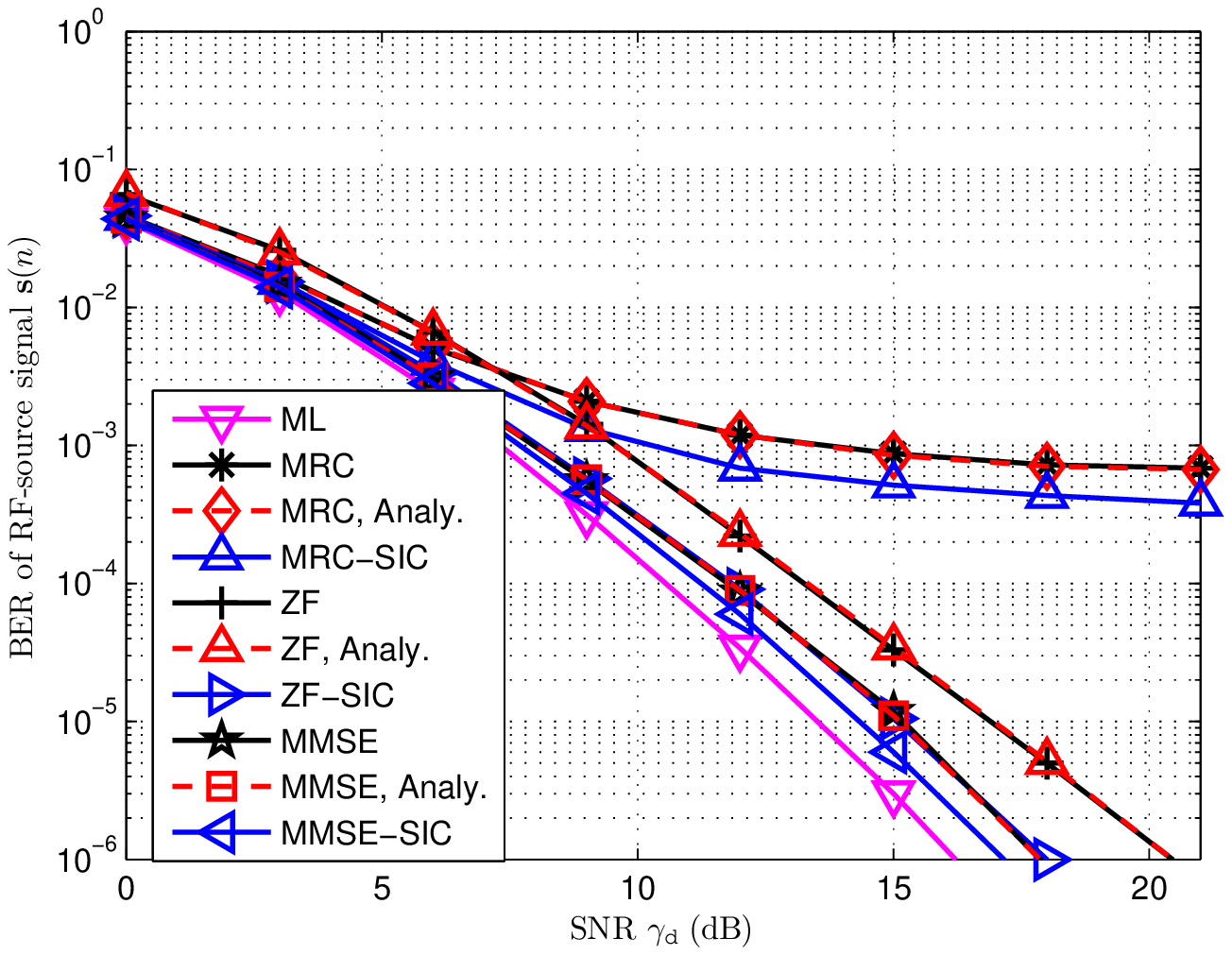}}
\subfigure[BERs of A-BD signal $c(n)$] {\includegraphics[height=2.5in,width=3.3in,angle=0]{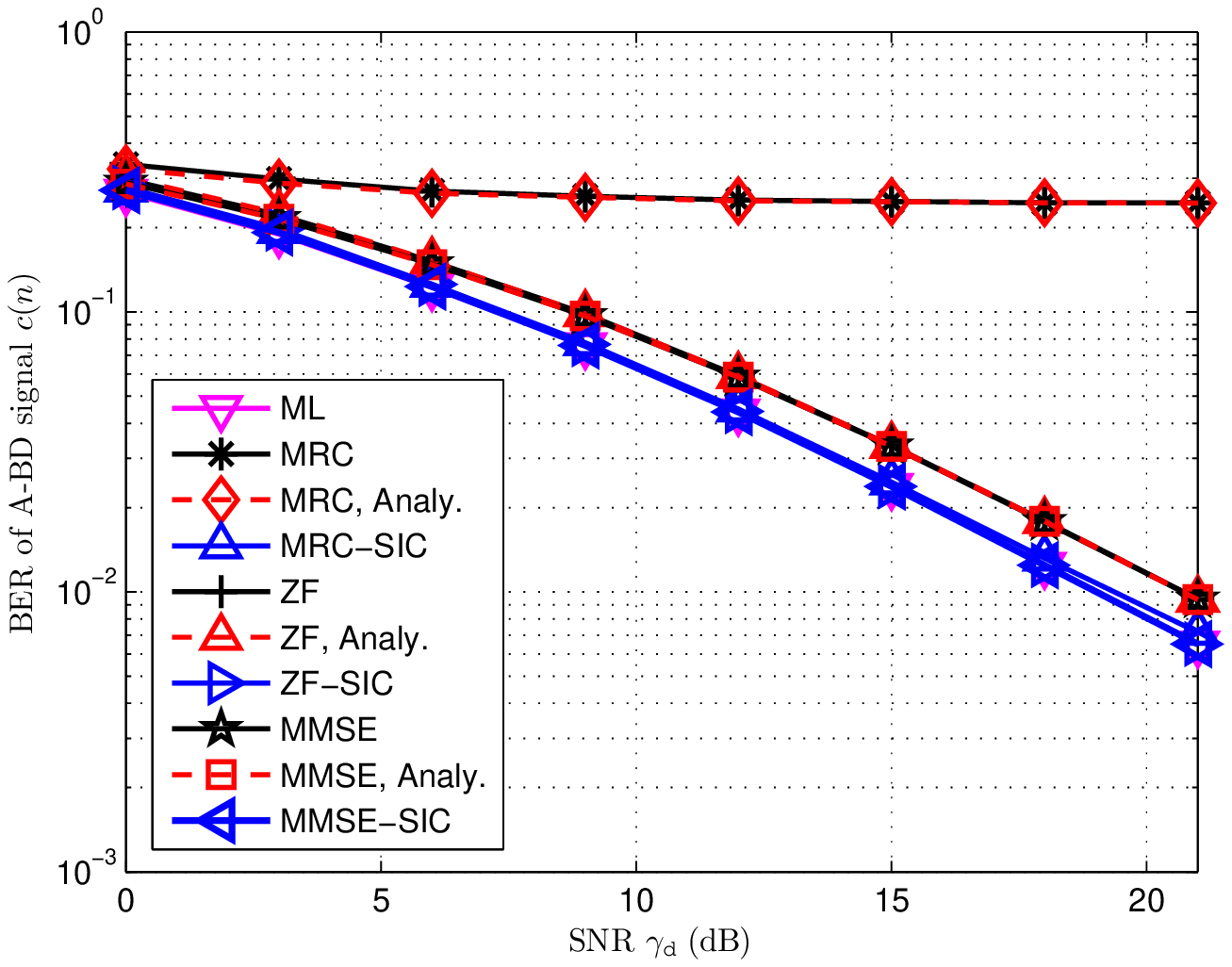}}
\caption{{{\textcolor{black}{BERs of $\bs(n)$ and $c(n)$ by using all detectors, for $K=1$ and $\Delta \gamma=-10$ dB, under flat fading channel}.}}}
\label{fig:Fig3}
\end{figure}


From Fig.~\ref{fig:Fig3}(b), we have two observations. First, the performance of MRC detector is poor, and the SIC-based ZF and MMSE detectors have a lower BER than the conventional ZF and MMSE detectors. In particular, at a BER level of $10^{-2}$, the SIC-based detectors achieve an SNR gain of 1.7 dB compared to the ZF and MMSE detectors. Second, the proposed SIC-based detectors achieve almost the same BER performance as the optimal ML detector.


%


\subsection{Frequency-selective Fading Channel}\label{sec:simulations_ofdm}
For frequency-selective fading channels, we set the numbers of multi-paths $L_{\sf f}=L_{\sf v}=8$, and $L_{\sf g}=1$ due to short A-BD-to-C-RX distance in practice. For OFDM modulation, we set $N=64$ and $N_{\sf c}=16$.

%



\subsubsection{Case of zero delay at C-RX and perfect synchronization at A-BD}\label{sec:simulations_ofdm_1}
In this subsection, we consider the case in which the additional delay at the C-RX $d=0$ and the transmission of A-BD signal is perfectly synchronized to the arrival of OFDM signal at the A-BD.

Fig.~\ref{fig:ofdm_cs}(a) compares the BERs of the RF-source signal $\bs(n)$ for ML detection in the CABC system and the direct-link OFDM system without an A-BD. First, we observe that for the CABC system, the BER performance of $\bs(n)$ improves as the backscatter-link power increases (i.e., as the relative SNR $\Delta \gamma$ increases). Second, we observe that the BER performance of $\bs(n)$ for the CABC system is in general better than that for the direct-link OFDM system, even for the case in which the backscatter-link power is only 10\% of the direct-link power (i.e., $\Delta \gamma=-10$ dB). This observation verifies that the existence of an A-BD can enhance the ML detection performance of the RF-source signal, especially for the scenario of higher backscatter-link power. Third, we observe that the simulated BERs of $\bs(n)$ coincide with the analytical BERs which are computed by using Theorem~\ref{theorem:BERML_ofdm} and the simulated BER of $c(n)$ shown in Fig.~\ref{fig:ofdm_cs}(b). Thus this verifies Theorem~\ref{theorem:BERML_ofdm}.

\begin{figure}[!t]
\centering
\subfigure[BERs of RF-source signal $\bs(n)$] {\includegraphics[height=2.5in,width=3.3in,angle=0]{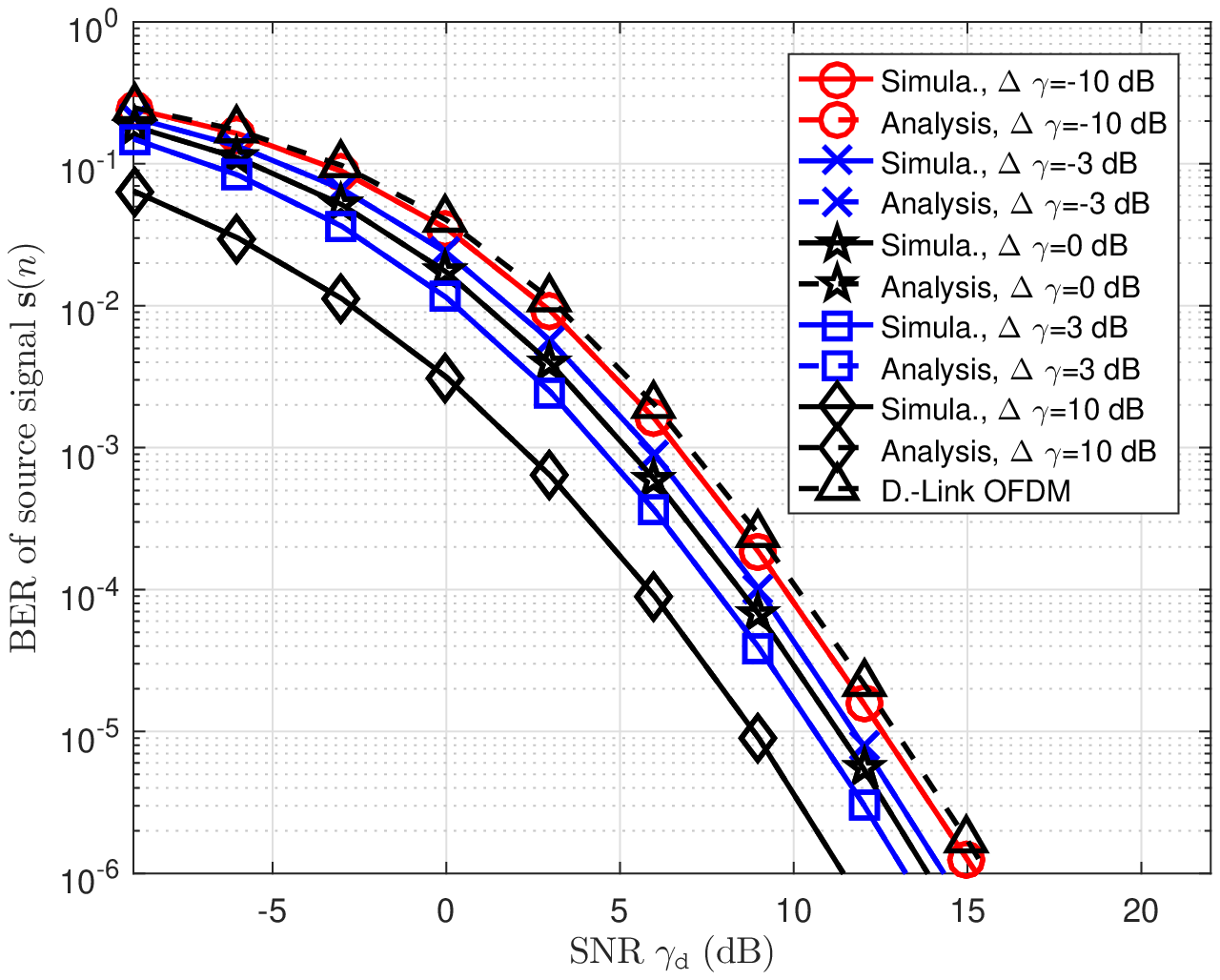}}
\subfigure[BERs of A-BD signal $c(n)$] {\includegraphics[height=2.5in,width=3.3in,angle=0]{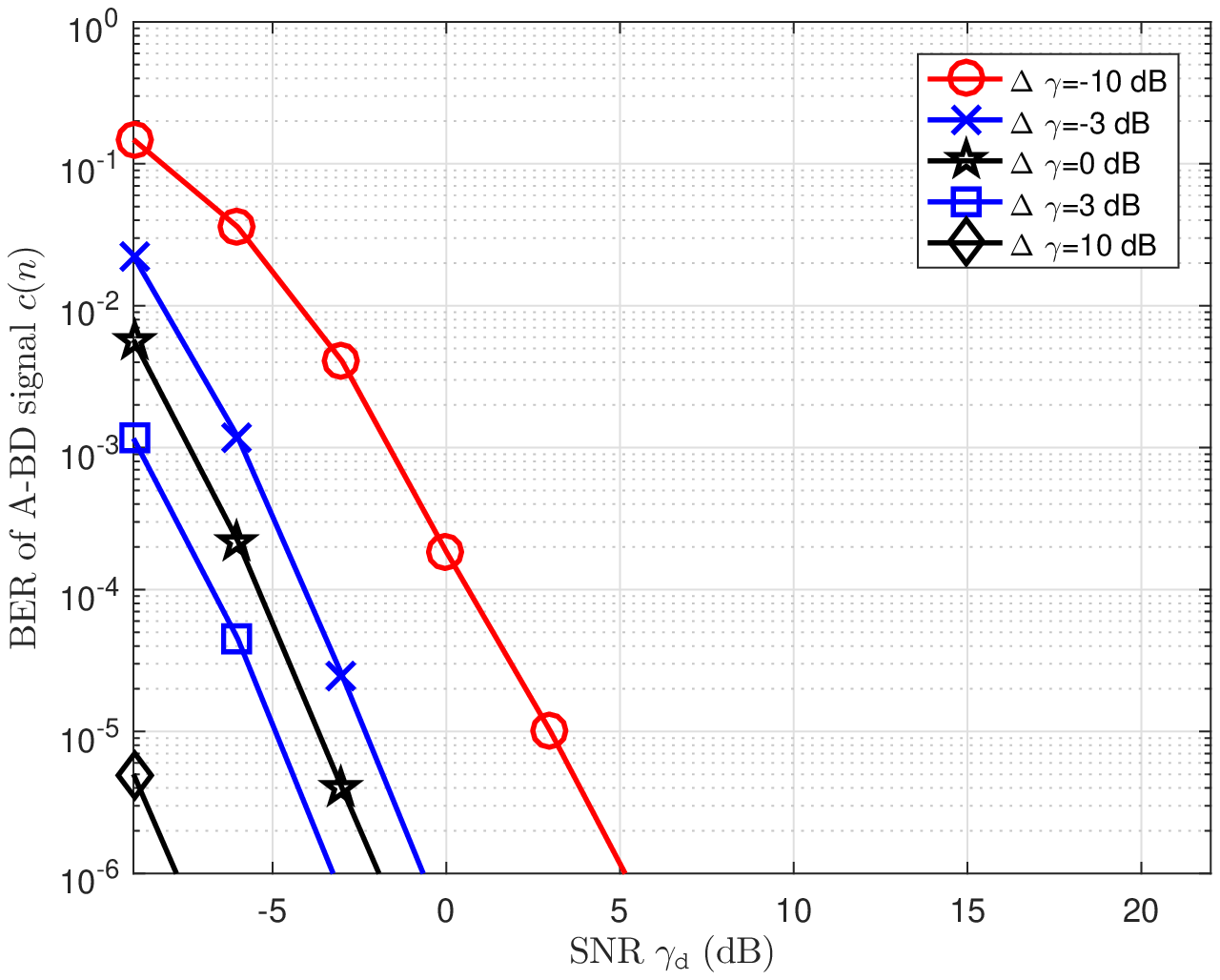}}
\caption{BERs of $\bs(n)$ and $c(n)$ for different $\Delta \gamma$'s, {{\textcolor{black}{under frequency-selective fading channel}}}.}
\label{fig:ofdm_cs}
\end{figure}

Fig.~\ref{fig:ofdm_cs}(b) compares the BERs of the A-BD signal $c(n)$ for the CABC system with different relative SNRs $\Delta \gamma$'s. First, we observe that the BER of $c(n)$ decreases very quickly as the SNR $\gamma_{\sf d}$ increases, compared to that in Fig.~\ref{fig:Fig2Diffk}(b) with flat fading channels. For instance, at a BER level of $10^{-2}$, the BER performance achieves an SNR gain of around 22 dB for $\Delta \gamma =-10$ dB, compared to the BER for $K=1$ with flat fading channels (i.e., the red solid line with circle maker  in Fig.~\ref{fig:Fig2Diffk}(b). This can be explained as the spreading gain and diversity gain with frequency-selective fading channels. Second, we observe that the BER is lower for higher backscatter-link power (i.e., larger relative SNR $\Delta \gamma$).

\subsubsection{Case of delay at C-RX and imperfect synchronization at A-BD}\label{sec:simulations_ofdm_2}
In this subsection, we consider the case in which the additional delay at the C-RX $d>0$ and the transmission of the A-BD signal is delayed by $d_0 \ ( d_0 \geq 0)$ compared to the arrival of OFDM signal at the A-BD.

Fig.~\ref{fig:ofdm_cs_wdelay} compares the BERs of $\bs(n)$ and $c(n)$ for ML detection in the CABC system. First, we consider the scenario of $d_0=0$. For $d \leq d_{\max}=9$, the BER performances of both $\bs(n)$ and $c(n)$ are the same as those for the case of $d =0$; while for $d > d_{\max}$, the BERs are higher than those for the case of $d =0$. These observation verify that when the CP length $N_{\sf c}$ is sufficient to cover the total channel delay, the proposed ML detector for CABC system is optimal; but when the $N_{\sf c}$ is not such sufficient, the BER performances suffer from both the IBI and its resulting ICI due to a loss of orthogonality among subcarriers. Then, we investigate the scenario of $d_0>0$. Compared to the scenario of $d_0=0$, the BER performance of $\bs(n)$ is obviously worsen, but the BER of $c(n)$ is only increased slightly. The robustness of detecting $c(n)$ comes from the spreading gain and diversity gain with frequency-selective fading channels.

\begin{figure}[!t]
\centering
\subfigure[BERs of RF-source signal $\bs(n)$] {\includegraphics[height=2.5in,width=3.3in,angle=0]{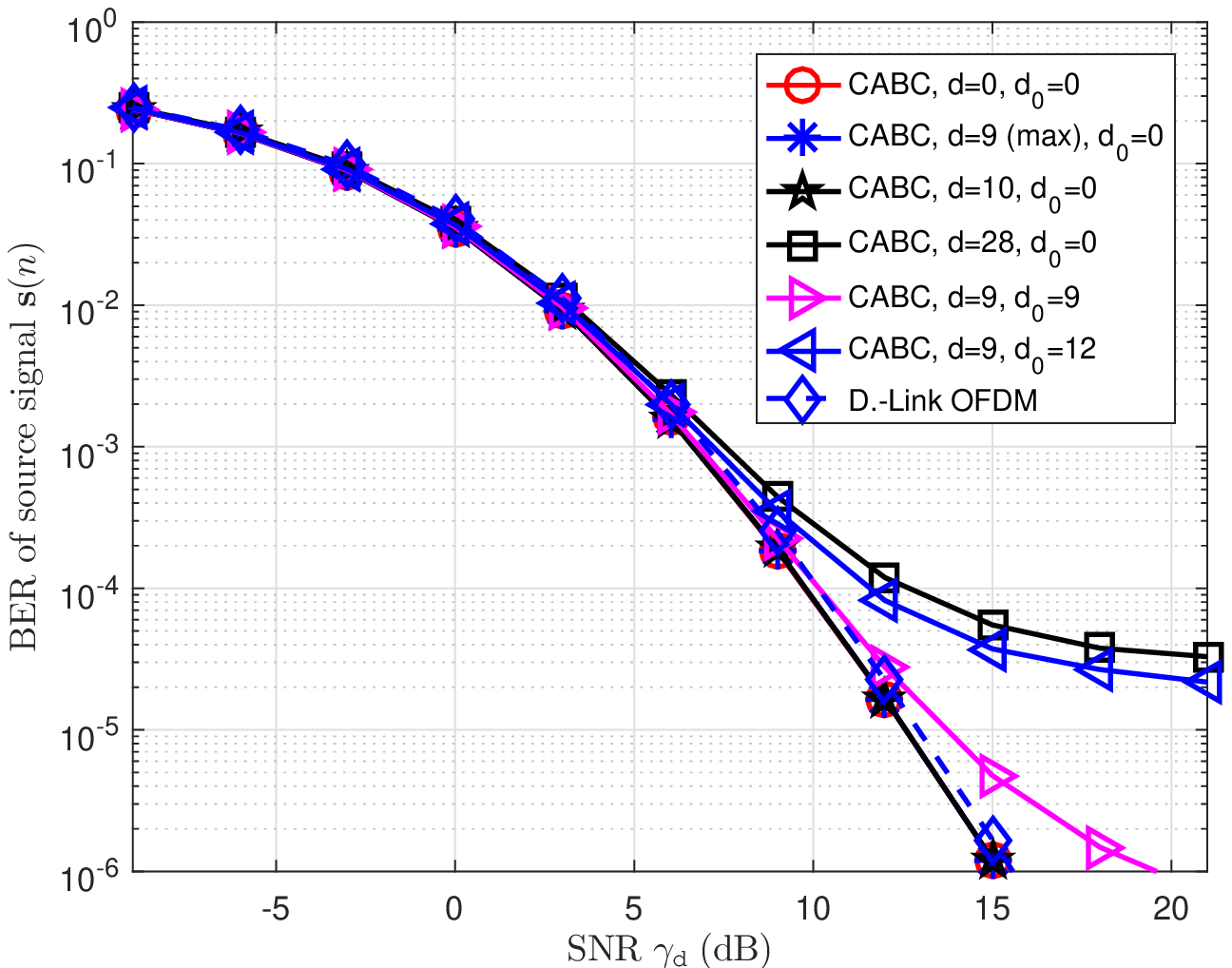}}
\subfigure[BERs of A-BD signal $c(n)$] {\includegraphics[height=2.5in,width=3.3in,angle=0]{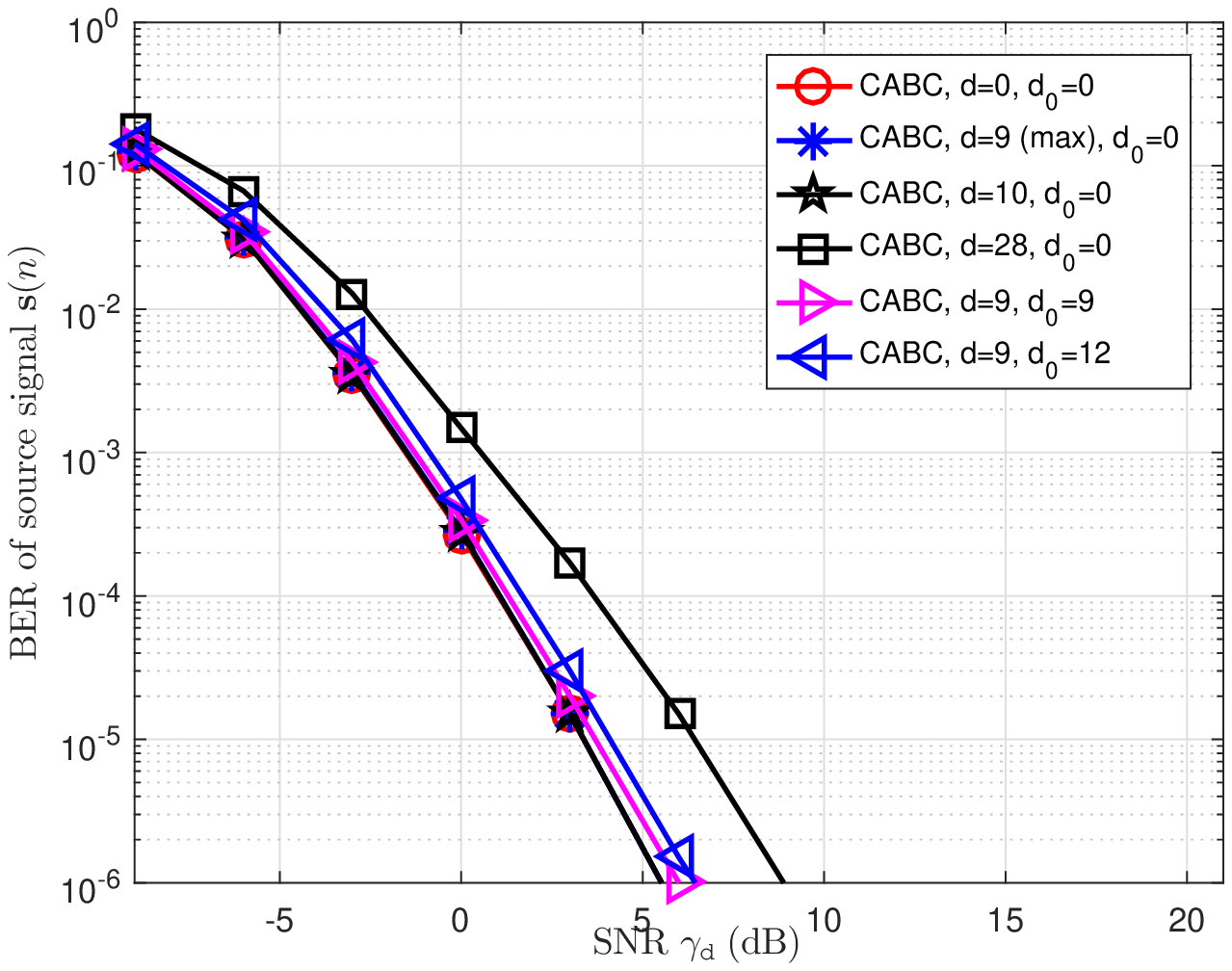}}
\caption{BERs of $\bs(n)$ and $c(n)$ for cases with delays, {{\textcolor{black}{under frequency-selective fading channel}}}.}
\label{fig:ofdm_cs_wdelay}
\end{figure}

\section{Conclusions}\label{conslusion}
This paper has proposed a novel cooperative AmBC (CABC) system, in which a cooperative receiver is designed to recover signals from both the RF source and the A-BD. For flat fading channels, by exploiting the structural property of the derived system model, we have proposed the optimal ML detector, suboptimal linear detectors and SIC-based detectors. For frequency-selective fading channels, we have proposed a low-complexity optimal ML detector for the CABC system over ambient OFDM carriers. Moreover, for both kinds of fading channels, we have obtained the BER expressions in closed forms for the proposed detectors. Numerical results have verified that, when the A-BD and RF-source signals have equal symbol period, the proposed SIC-based detectors can achieve near-ML detection performance for typical application scenarios in which the backscattered signal power is lower than the RF-source signal power; and when the A-BD symbol period is longer than the RF-source symbol period, the existence of backscattered signal in the CABC system can significantly enhance the ML detection performance of the RF-source signal compared to a conventional SIMO system without an A-BD, and the detection performance of the A-BD signal benefits from the spread spectrum gain due to longer A-BD symbol period. {{\textcolor{black}{The designed CABC system has a great potential for applications in future green IoT systems like smart homes and wearable sensor networks.}}}

\appendices
\section{Proof of Theorem~\ref{theorem:BERML}}\label{proof_Thm_ML}
\begin{proof}
Denote $\hats_0(n)$ and $\hatc_0(n)$ as the ML solutions of $s(n)$ and $c(n)$ for \eqref{eq:yn_1}, respectively. Since $\hats_0(n)$ and $\hatc_0(n)$ are mutually related based on \eqref{eq:ML_flat} for $K=1$, we discuss the estimation of $\hats_0(n)$ from $\hatc_0(n)$, and the estimation of $\hatc_0(n)$ from $\hats_0(n)$, respectively.

\subsubsection{Estimating $\hats_0(n)$ from $\hatc_0(n)$}

From the received signal model in \eqref{eq:receives yk} for $K=1$, we construct a sufficient statistic for detecting $s(n)$ from $\by(n)$ as follows
\begin{align}
  z(n) &= \frac{\tilbh^H |_{\hatc_0(n)}}{\|\tilbh |_{\hatc_0(n)}\|^2} \by (n) =\frac{(\bh_1 + \bh_2 \hatc_0(n))^H}{\|\bh_1 +\bh_2 \hatc_0(n)\|^2} \by (n).
\end{align}
With $\hatc_0(n)$, the ML estimate $\hats_0 (n)$ is the quantized output of $z(n)$.


If $\hatc_0(n) = c(n)$, we have
 \begin{align}
 z(n) = s(n) + \tilu(n),
\end{align}
with $\tilu(n) \sim \calC \calN(0,\frac{\sigma^2}{\| \bh_1 +\bh_2\|^2})$ for $c(n) = 1$, and $\tilu(n) \sim \calC \calN(0,\frac{\sigma^2}{\| \bh_1 -\bh_2\|^2})$ for $c(n) = -1$, which gives the following BER of $s(n)$
\begin{align}
  a_1 (\bH) \!= \! \frac{1}{2} Q \left( \! \! {\frac{ \| \bh_1 \!+\! \bh_2 \|}{\sigma}} \! \! \right) \!+\!  \frac{1}{2} Q \left( \! \! {\frac{ \| \bh_1 \!-\! \bh_2 \|}{\sigma}} \! \! \right).  \label{eq:aa1}
  \end{align}

Next, we consider the case that $c(n)=1$ and $\hatc_0(n)=-1$. In this case, we have
\begin{align}
  z_1(n) = \theta_1 s(n) + \tilu_1(n),
\end{align}
where $\theta_1=\frac{(\bh_1-\bh_2)^H (\bh_1+\bh_2)}{\| \bh_1 -\bh_2\|^2}$ and $\tilu_1(n) \sim \calC \calN(0,\frac{\sigma^2}{\| \bh_1 -\bh_2\|^2})$.
Let $j=\sqrt{-1}$, and denote $\theta_1 = \theta_{R,1} + \theta_{I, 1} j$, $s(n)=s_R(n) + s_I(n) j$, $\tilu_1(n)=\tilu_{R,1}(n) + \tilu_{I,1}(n) j$, and $z_1(n)=z_{R, 1}(n) + z_{I,1}(n) j$. Denote the ML estimation $\hats_0(n)$ as $\hats_0(n)=\hats_{0, R}(n) + \hats_{0, I}(n) j$. Note that $z_{R, 1}(n) \sim \calN(\theta_{R, 1} s_R - \theta_{I,1} s_I, \frac{\sigma^2}{2\| \bh_1 -\bh_2\|^2})$, and $z_{I,1}(n) \sim \calN(\theta_{R,1} s_I + \theta_{I, 1} s_R, \frac{\sigma^2}{2 \| \bh_1 -\bh_2\|^2})$. Hence, the BER for $s_R(n)$ is given by~\eqref{eq:BER_Ps_R_Pc} at the top of the next page.
\begin{figure*}[t!]
\begin{align}\label{eq:BER_Ps_R_Pc}
  P_{{\sf s}, R,1} &= \Pr\left(s_R = \frac{\sqrt{2}}{2}\right) \Pr\left(s_I = \frac{\sqrt{2}}{2}\right) \Pr\left( \frac{\sqrt{2}}{2} \theta_{R,1} - \frac{\sqrt{2}}{2} \theta_{I,1} +\tilu_{R,1} < 0 \right) +... \nonumber \\
  &\qquad \Pr\left(s_R = \frac{\sqrt{2}}{2}\right) \Pr\left(s_I = -\frac{\sqrt{2}}{2}\right) \Pr\left( \frac{\sqrt{2}}{2} \theta_{R,1} + \frac{\sqrt{2}}{2} \theta_{I,1} +\tilu_{R,1} < 0 \right)+... \nonumber \\
  &\qquad \Pr\left(s_R = -\frac{\sqrt{2}}{2}\right) \Pr\left(s_I = \frac{\sqrt{2}}{2}\right) \Pr\left( -\frac{\sqrt{2}}{2} \theta_{R,1} - \frac{\sqrt{2}}{2} \theta_{I,1} +\tilu_{R,1} > 0 \right) +... \nonumber \\
  &\qquad \Pr\left(s_R = -\frac{\sqrt{2}}{2}\right) \Pr\left(s_I = -\frac{\sqrt{2}}{2}\right) \Pr\left( -\frac{\sqrt{2}}{2} \theta_{R,1} + \frac{\sqrt{2}}{2} \theta_{I,1} +\tilu_{R,1} > 0 \right) \nonumber \\
  &=\frac{1}{2}\calQ \left( \frac{\sqrt{ \| \bh_1 -\bh_2\|^2}(\theta_{R,1} + \theta_{I,1})}{\sigma}\right) +\frac{1}{2}\calQ \left( \frac{\sqrt{\| \bh_1 -\bh_2\|^2}(\theta_{R,1} - \theta_{I,1})}{\sigma}\right).
\end{align}
\end{figure*}


It can be checked that the BER for $s_{I,1}(n)$ is the same as that for $s_{R,1}(n)$. Thus the error rate of $s(n)$ for the case $c(n) = 1$ and $\hat{c}_0(n) = -1$ is given by
\begin{align}\label{eq:BER_Ps_1}
  P_{{\sf e, s}, 1} &= \frac{1}{2}\calQ \left( \frac{\| \bh_1 -\bh_2\| (\theta_{R£¬1} + \theta_{I£¬1})}{\sigma}\right) +... \nonumber \\
  &\quad \frac{1}{2}\calQ \left( \frac{\| \bh_1 -\bh_2\|(\theta_{R£¬1} - \theta_{I£¬1})}{\sigma}\right).
\end{align}

Similarly, for the case that $c(n) = -1$ and $\hat{c}_0(n) = 1$, the error rate for $s(n)$ is given by
\begin{align}\label{eq:BER_Ps_2}
  P_{{\sf e, s}, 2} &= \frac{1}{2}\calQ \left( \frac{\| \bh_1 +\bh_2\|(\theta_{R,2} + \theta_{I,2})}{\sigma}\right) +...\nonumber \\
  &\quad \frac{1}{2}\calQ \left( \frac{\| \bh_1 +\bh_2\|(\theta_{R,2} - \theta_{I,2})}{\sigma}\right),
\end{align}
where
$\theta_{R,2} = \textrm{Re} \left\{ \frac{(\bh_1+\bh_2)^H (\bh_1-\bh_2)}{\| \bh_1+\bh_2\|^2} \right\}$ and $\theta_{I,2} = \textrm{Im} \left( \frac{(\bh_1+\bh_2)^H (\bh_1-\bh_2)}{\| \bh_1+\bh_2\|^2} \right)$.

From~\eqref{eq:BER_Ps_1} and~\eqref{eq:BER_Ps_2}, for the case of $\hatc_0(n) \neq c(n)$, the BER, $a_2(\bH)$, of $s(n)$ is thus given in~\eqref{eq:a2}.

Considering both cases of $\hatc_0(n) = c(n)$ and $\hatc_0(n) \neq c(n)$, the BER of $s(n)$ is thus
\begin{align}\label{eq:BER_Ps_Pc}
  P_{\sf e, s} (\bH) &= (1-P_{\sf e, c} (\bH) ) a_1 (\bH) + P_{\sf e, c} (\bH) a_2 (\bH).
\end{align}

\subsubsection{Estimating $\hatc_0(n)$ from $\hats_0(n)$}
With $\hats_0(n)$, the ML estimate of $c(n)$ is as follows
 \begin{align} \label{eq:ML_c}
  \hatc_0 (n) = \underset{c(n) \in \calA_{\sf c}}{\arg \min} \left\| \by (n) - \bh_1 \hats(n) -\bh_2 \hats (n) c(n) \right\|.
\end{align}
Let
\begin{align}
  z(n) =\textrm{Re} \left\{ \frac{\bh_2^H \hats_0^*(n)}{\|\bh_2\|^2} \left( \by-\bh_1 \hats_0(n) \right)\right\}.
\end{align}
The ML detection yields $\hatc_0(n)=1$ if $z(n) >0$, and $\hatc_0(n)=-1$ if $z(n) \leq 0$. We further have
\begin{align}
  z(n) &= \textrm{Re} \left\{ \hats_0^{\ast} (n) s(n) \right\} c(n) +... \nonumber \\
  &\qquad \textrm{Re} \left\{ \frac{\bh_2^H \bh_1}{\|\bh_2\|^2} \left( \hats_0^{\ast}(n) s(n) -1 \right)\right\} + \tilu(n),
\end{align}
where and the noise $\tilu (n)= \textrm{Re} \left\{ \frac{\bh_2^H s^H(n) \bu}{\|\bh_2\|^2}\right\}$, and $\tilu (n) \sim \calN(0, \frac{\sigma^2}{2 \| \bh_2\|^2})$.

If $\hats_0(n)=s(n)$, we have $z(n)=c(n)+ \tilu(n)$. The equivalent SNR of estimating $c(n)$ from $z(n)$ is $\frac{2  \| \bh_2\|^2}{\sigma^2}$, which gives the BER, $b_1(\bH)$, of $c(n)$ shown in \eqref{eq:b1}.


If $\hats_0(n) \neq s(n)$, the product $\hats_0^{\ast} (n) s(n)$ takes the value of $j, \ -j$ and $-1$, termed as case $i =1, 2, 3$, with probability of $\sqrt{2}-1$, $\sqrt{2}-1$ and $(\sqrt{2}-1)^2$, respectively. Clearly, for each case $i$, the $z(n)$ follows Gaussian distribution with variance $\frac{\sigma^2}{2 \| \bh_2\|^2}$ and mean $\mu_i$ which is given by
\begin{align}
  \mu_1 &= -\textrm{Im} \left\{ \frac{\bh_2^H \bh_1}{\|\bh_2\|^2}\right\}-\textrm{Re} \left\{ \frac{\bh_2^H \bh_1}{\|\bh_2\|^2}\right\}, \nonumber \\
  \mu_2 &= \textrm{Im} \left\{ \frac{\bh_2^H \bh_1}{\|\bh_2\|^2}\right\}-\textrm{Re} \left\{ \frac{\bh_2^H \bh_1}{\|\bh_2\|^2}\right\}, \nonumber \\
  \mu_3 &= -c(n) - 2 \textrm{Re} \left\{ \frac{\bh_2^H \bh_1}{\|\bh_2\|^2}\right\}. \nonumber
\end{align}

We assume that QPSK modulation adopts typical Gray mapping. For case 1 and 2, only one bit of each symbol is decoded wrongly, while for case 3, both bits are decoded wrongly. It is noted that the means for case 1 and case 2 are independent of $c(n)$, the error rate of $c(n)$ is thus $0.5$, which implies that the error rates for both case 1 and 2 are $0.5 P_{\sf e, s} (\bH) \left(1- P_{\sf e, s} (\bH)\right)$. Moreover, it can be checked that the error rate of $c(n)$ for case 3, denoted by $b_2 (\bH)$, is given in~\eqref{eq:b2}.

Hence, the BER of $c(n)$ can be written in terms of $P_{\sf e, s} (\bH)$, as
\begin{align}\label{eq:BER_Pc_Ps}
  P_{\sf e, c} (\bH) &= (1- P_{\sf e, s} (\bH))^2 b_1 (\bH) +... \nonumber \\
  &\quad P_{\sf e, s} (\bH) \left(1- P_{\sf e, s} (\bH)\right) + P_{\sf e, s}^2 (\bH)  b_2 (\bH).
\end{align}

From~\eqref{eq:BER_Ps_Pc} and~\eqref{eq:BER_Pc_Ps}, we have the quadratic function
\begin{align}
  C_2(\bH) P_{\sf e, s}^2(\bH) + C_1(\bH) P_{\sf e, s}(\bH) +C_1(\bH)=0.
\end{align}
From the fact that $0 \leq C_0(\bH) <1, \ C_1(\bH) < 0, \ C_2(\bH)<0$, the obtained the BER of $s(n)$ in~\eqref{eq:BER_s_analysis}, and thus the BER of $c(n)$ in~\eqref{eq:BER_t_analysis}. This completes the proof.
\end{proof}

\section{Proof of Proposition~\ref{theorem:BERMRC}}\label{proof_Thm_MRC}
\begin{proof}
{\em 1) BER of $s(n)$}: With MRC detection, the estimated RF-source signal is given as
\begin{align}
  \tils(n) &=\frac{\bh_1^H}{\| \bh_1\|^2} \by = s(n) + \frac{\bh_1^H \bh_2}{\| \bh_1\|^2} s(n) c(n) + \tilu(n),
\end{align}
where $\tilu(n)=\frac{\bh_1^H \bu(n)}{\| \bh_1\|^2}$, and $\tilbu(n) \sim \calC\calN\left(0, \frac{\sigma^2}{\|\bh_1\|^2}\right)$. Denote $\tils(n)=\tils_{R}(n) + \tils_{I}(n) j$ and $\tilu(n)=\tilu_{R}(n) + \tilu_{I}(n) j$. We have
\begin{align}
  &\tils_R(n) =  \\
  &s_R(n) \!+\! \frac{\text{Re} \{\bh_1^H \bh_2\} s_R(n) \!-\! \text{Im} \{\bh_1^H \bh_2\} s_I(n)}{\| \bh_1\|^2} c(n) \!+\! \tilu_R(n), \nonumber
\end{align}
where $\tilbu_R(n) \sim \calC\calN\left(0, \frac{\sigma^2}{2 \|\bh_1\|^2}\right)$.

We first consider four cases that $s_R(n)=\frac{1}{\sqrt{2}}$, $s_I(n)=\pm \frac{1}{\sqrt{2}}$ and $c(n)=\pm 1$, each case with probability of $\frac{1}{8}$. For each case, the $\tils_R(n)$ follows Gaussian distribution with variance $\frac{\sigma^2}{2 \| \bh_1\|^2}$ and mean $\mu_i$ which is given by
we have that
\begin{align}
  \mu_1 &= \frac{1}{\sqrt{2}} \left( 1+ \frac{\textrm{Re} \left\{\bh_1^H\bh_2\right\}}{\| \bh_1\|^2}-\frac{\textrm{Im} \left\{\bh_1^H\bh_2\right\}}{\| \bh_1\|^2}\right), \quad \nonumber \\
  &\qquad \text{for} \ s_I(n)= \frac{1}{\sqrt{2}}, \ c(n)=1 \nonumber \\
    \mu_2 &= \frac{1}{\sqrt{2}} \left( 1+ \frac{\textrm{Re} \left\{\bh_1^H\bh_2\right\}}{\| \bh_1\|^2}+\frac{\textrm{Im} \left\{\bh_1^H\bh_2\right\}}{\| \bh_1\|^2}\right), \quad \nonumber \\
  &\qquad \text{for} \ s_I(n)= -\frac{1}{\sqrt{2}}, \ c(n)=1  \nonumber \\
      \mu_3 &= \frac{1}{\sqrt{2}} \left( 1- \frac{\textrm{Re} \left\{\bh_1^H\bh_2\right\}}{\| \bh_1\|^2} + \frac{\textrm{Im} \left\{\bh_1^H\bh_2\right\}}{\| \bh_1\|^2}\right), \quad \nonumber \\
  &\qquad \text{for} \ s_I(n)= \frac{1}{\sqrt{2}}, \ c(n)=-1  \nonumber \\
        \mu_4 &= \frac{1}{\sqrt{2}} \left( 1- \frac{\textrm{Re} \left\{\bh_1^H\bh_2\right\}}{\| \bh_1\|^2}-\frac{\textrm{Im} \left\{\bh_1^H\bh_2\right\}}{\| \bh_1\|^2}\right), \quad \nonumber \\
  &\qquad \text{for} \ s_I(n)= -\frac{1}{\sqrt{2}}, \ c(n)=-1 .
\end{align}
The error probability is thus given in~\eqref{eq:ber_real_mrc} at the top of the next page.
\begin{figure*}
\begin{align}
  &\Pr \Big(\hats_R(n)=-\frac{1}{\sqrt{2}}, s_R(n)=  \frac{1}{\sqrt{2}}\Big) \nonumber \\
  &=\frac{1}{8} \! \left[ \!\Pr \! \left( \! \tils_R(n) \! < \! 0 \Big|  s_R(n) \!=\! \frac{1}{\sqrt{2}}, s_I(n) \! =\!  \frac{1}{\sqrt{2}}, c(n) \!=\! 1\right) \!+\!
  \Pr \! \left( \! \tils_R(n) \!<\! 0 \Big|  s_R(n) \!=\! \frac{1}{\sqrt{2}}, s_I(n) \!=\!  -\frac{1}{\sqrt{2}}, c(n)\! =\!1 \!\right) \! \right] \!+\!... \nonumber \\
  &\quad \frac{1}{8} \! \left[ \! \Pr \! \left( \! \tils_R(n) \! <\! 0 \Big|  s_R(n) \!=\! \frac{1}{\sqrt{2}}, s_I(n) \!=\! \frac{1}{\sqrt{2}}, c(n) \!=\! -1 \! \right) \!+\! \Pr \! \left(\! \tils_R(n) \!<\! 0 \Big|  s_R(n) \!=\! \frac{1}{\sqrt{2}}, s_I(n) \!=\! -\frac{1}{\sqrt{2}}, c(n) \!=\! -1\! \right) \! \right]\nonumber \\
&= \frac{1}{8} \left[ \calQ \left( \frac{\| \bh_1 \|}{\sigma} \left( 1+ \frac{\textrm{Re} \left\{\bh_1^H\bh_2\right\}-\textrm{Im} \left\{\bh_1^H\bh_2\right\}}{\| \bh_1\|^2}\right)\right)+\calQ \left( \frac{\| \bh_1 \|}{\sigma} \left( 1+ \frac{\textrm{Re} \left\{\bh_1^H\bh_2\right\}+\textrm{Im} \left\{\bh_1^H\bh_2\right\}}{\| \bh_1\|^2}\right)\right) \right] +... \nonumber \\
  &\quad \frac{1}{8} \left[\calQ \left( \frac{\| \bh_1 \|}{\sigma} \left( 1- \frac{\textrm{Re} \left\{\bh_1^H\bh_2\right\}-\textrm{Im} \left\{\bh_1^H\bh_2\right\}}{\| \bh_1\|^2}\right)\right)+ \calQ \left( \frac{\| \bh_1 \|}{\sigma} \left( 1- \frac{\textrm{Re} \left\{\bh_1^H\bh_2\right\}+\textrm{Im} \left\{\bh_1^H\bh_2\right\}}{\| \bh_1\|^2}\right)\right) \right].\label{eq:ber_real_mrc}
\end{align}
\end{figure*}

Then for the four cases that $s_R(n)=-\frac{1}{\sqrt{2}}$, $s_I(n)=\pm \frac{1}{\sqrt{2}}$ and $c(n)=\pm 1$, each case with probability of $\frac{1}{8}$, it can be checked that the error probability $\Pr \left(\hats_R(n)=\frac{1}{\sqrt{2}}, \ s_R(n)=-\frac{1}{\sqrt{2}}\right)$ is equal to $\Pr \left(\hats_R(n)=-\frac{1}{\sqrt{2}}, \ s_R(n)=\frac{1}{\sqrt{2}}\right)$ which is given in~\eqref{eq:ber_real_mrc}. Hence, the error rate expression of the real part $s_R(n)$ is the same as the right-hand-side of \eqref{eq:BER_s_analysis_mrc}. From the symmetry property of real part and imaginary part of QPSK modulation, the BER of $s(n)$ is obtained in \eqref{eq:BER_s_analysis_mrc}.

{\em 2) BER of $c(n)$}: With MRC detection, the estimated backscattered signal is written as follows
\begin{align}
  \tilx_2(n) &\!=\!\frac{\bh_2^H}{\| \bh_2\|^2} \by \!=\! s(n) c(n) \!+\! \frac{\bh_2^H \bh_1}{\| \bh_2\|^2} s(n) \!+\! \tilu(n),
\end{align}
where $\tilu(n)=\frac{\bh_2^H \bu(n)}{\| \bh_2\|^2}$, and $\tilu(n) \sim \calC\calN\left(0, \frac{\sigma^2}{\|\bh_2\|^2}\right)$. Define the following signal
\begin{align}
  z(n) &= \text{Re} \left\{ \frac{\tilx_2(n)}{\hats(n)}\right\} \\
  &= \text{Re} \left\{ \frac{s(n)}{\hats(n)}\right\} c(n) + \text{Re} \left\{ \frac{\bh_2^H \bh_1}{\| \bh_2\|^2}  \frac{s(n) }{\hats(n)}\right\} +\tilu^{\sharp}(n),\nonumber
\end{align}
where $\tilu^{\sharp}(n)=\frac{\tilu(n)}{\hats(n)}$, and $\tilu^{\sharp}(n) \sim \calC\calN\left(0, \frac{\sigma^2}{\|\bh_2\|^2}\right)$. Denote $\tilu^{\sharp}(n)=\tilu^{\sharp}_R(n)+j\tilu^{\sharp}_I(n)$.

For the case $\hats(n)=s(n)$, we have $z(n)=c(n)+\text{Re} \left\{ \frac{\bh_2^H \bh_1}{\| \bh_2\|^2}\right\}+\tilu^{\sharp}_R(n)$. Following similar steps as in the proof of Theorem 1, we obtain the BER of $c(n)$ for the case $\hats(n)=s(n)$ as follows
\begin{align}
  &\Pr \left(\hatc(n) \neq c(n), \hats(n) = s(n)\right) \label{eq:ber_c1_mrc}\\
  &= \left(1 \!-\!  P_{\sf e, s} (\bH)\right)^2 \! \Big[ \! \Pr(c(n) \!=\! 1) \Pr(z(n)\!<\!0 | c(n)\!=\!1) \!+\!... \nonumber \\
  &\quad \Pr(c(n)=-1) \Pr(z(n)>0 | c(n)=-1)\Big] \nonumber \\
  &\!=\! \left(1 \!-\! P_{\sf e, s} (\bH)\right)^2 \! \Bigg[ \! \frac{1}{2} \calQ \! \left( \! \frac{\sqrt{2} \| \bh_2 \|}{\sigma} \left(\! 1 \!+\! \frac{\textrm{Re} \left\{\bh_2^H\bh_1\right\}}{\| \bh_2\|^2} \! \right) \! \right)\!+\!... \nonumber \\
  &\quad \frac{1}{2} \calQ \left( \frac{\sqrt{2} \| \bh_2 \|}{\sigma} \left( 1- \frac{\textrm{Re} \left\{\bh_2^H\bh_1\right\}}{\| \bh_2\|^2}\right)\right) \! \Bigg].\nonumber
\end{align}

If $\hats(n) \neq s(n)$, the fraction $\frac{s(n)}{\hats (n)}$ takes the value of $j, \ -j$ and $-1$, termed as case $i =1, 2, 3$, with probability of $\sqrt{2}-1$, $\sqrt{2}-1$ and $(\sqrt{2}-1)^2$, respectively. Clearly, for each case $i$, the $z(n)$ follows Gaussian distribution with variance $\frac{\sigma^2}{2 \| \bh_2\|^2}$ and mean $\mu_1 = -\textrm{Im} \left\{ \frac{\bh_2^H \bh_1}{\|\bh_2\|^2}\right\}$, $\mu_2 = \textrm{Im} \left\{ \frac{\bh_2^H \bh_1}{\|\bh_2\|^2}\right\}$, $\mu_3 = -c(n) - \textrm{Re} \left\{ \frac{\bh_2^H \bh_1}{\|\bh_2\|^2}\right\}$. Following similar steps as in the proof of Theorem 1, we obtain the BER of $c(n)$ for the case $\hats(n) \neq s(n)$ as follows
\begin{align}
  &\Pr \left(\hatc(n) \neq c(n), \hats(n) \neq s(n)\right) \nonumber \\
  &= \frac{P_{\sf e, s}^2 (\bH)}{2}\calQ \left( \frac{\sqrt{2} \| \bh_2 \|}{\sigma} \left( -1- \frac{\textrm{Re} \left\{\bh_2^H\bh_1\right\}}{\| \bh_2\|^2}\right)\right)+...\nonumber \\
  &\quad \frac{P_{\sf e, s}^2 (\bH)}{2} \calQ \left( \frac{\sqrt{2} \| \bh_2 \|}{\sigma} \left( -1+ \frac{\textrm{Re} \left\{\bh_2^H\bh_1\right\}}{\| \bh_2\|^2}\right)\right)+ ...\nonumber \\
  & \quad P_{\sf e, s} (\bH)\left(1-P_{\sf e, s} (\bH)\right).  \label{eq:ber_c2_mrc}
\end{align}
Thus, from~\eqref{eq:ber_c1_mrc} and~\eqref{eq:ber_c2_mrc}, we obtain the BER of $c(n)$ in \eqref{eq:BER_t_analysis_mrc}. This completes the proof.
\end{proof}

\section{Proof of Proposition~\ref{theorem:BERZF}}\label{proof_Thm_ZF}
\begin{proof}
With ZF detection, the estimated signal vector is written as follows
\begin{align}
  \tilbx(n) &= \left(\bH^H \bH \right)^{-1} \bH^H \by = \bx(n) + \tilbu(n),
\end{align}
where $\tilbu(n)=\left(\bH^H \bH \right)^{-1} \bH^H \bu(n)$, and $\tilbu(n) \sim \calC\calN\left(\bm{0}, \sigma^2 \bA\right)$.

Since $\tils(n)=\tilx_1(n)=s(n)+\tilu_1(n)$, the detecting SNR for $s(n)$ is thus $\sigma^2 A_{11}(\bH)$. Hence, the BER of $s(n)$ is given in~\eqref{eq:BER_s_analysis_zf}.

The estimated backscattered signal is written as follows
\begin{align}
  \tilx_2(n) &= s(n) c(n) + \tilu_2(n),
\end{align}
where $\tilu_2(n) \sim \calC\calN\left(0, \sigma^2 A_{22}(\bH)\right)$. Define the following signal
\begin{align}
  z(n) = \text{Re} \left\{ \frac{\tilx_2(n)}{\hats(n)}\right\} = \text{Re} \left\{ \frac{s(n)}{\hats(n)}\right\} c(n) + \tilu_2^{\sharp}(n),
\end{align}
where $\tilu^{\sharp}_2(n)=\frac{\tilu(n)}{\hats(n)}$, and $\tilu^{\sharp}_2(n) \sim \calC\calN\left(0, \frac{\sigma^2 A_{22}(\bH)}{2}\right)$. Denote $\tilu^{\sharp}_2(n)=\tilu^{\sharp}_{2, R}(n)+j\tilu^{\sharp}_{2,I}(n)$.

For the case $\hats(n)=s(n)$, we have $z(n)=c(n)+\tilu^{\sharp}_{2, R}(n)$. Following similar steps as in the proof of Theorem 1, we obtain the BER of $c(n)$ for the case $\hats(n)=s(n)$ as follows
\begin{align}
  &\Pr \left(\hatc(n) \neq c(n), \hats(n) = s(n)\right) \nonumber \\
  &= \left(1 \!-\! P_{\sf e, s} (\bH)\right)^2 \Pr(c(n) \!=\!1) \Pr(z(n) \!<\! 0 | c(n) \!=\! 1) \!+\! ...\nonumber \\
  &\quad  \left(1 \!-\! P_{\sf e, s} (\bH)\right)^2 \Pr(c(n)\!=\!1) \Pr(z(n)\!>\!0 | c(n)\!=\!-1) \nonumber \\
  &= \left(1-P_{\sf e, s} (\bH)\right)^2 \calQ \left( \frac{\sqrt{2}}{\sigma \sqrt{A_{22}(\bH)}}\right). \label{eq:ber_c1_zf}
\end{align}

If $\hats(n) \neq s(n)$, the fraction $\frac{s(n)}{\hats (n)}$ takes the value of $j, \ -j$ and $-1$, termed as case $i =1, 2, 3$, with probability of $\sqrt{2}-1$, $\sqrt{2}-1$ and $(\sqrt{2}-1)^2$, respectively. Clearly, for each case $i$, the $z(n)$ follows Gaussian distribution with variance $\frac{\sigma^2}{2 \| \bh_2\|^2}$ and mean $\mu_1 = \mu_2 = 0$, and $\mu_3 = -c(n)$. Following similar steps as in Appendix~\ref{proof_Thm_ML}, we obtain the BER of $c(n)$ for the case $\hats(n) \neq s(n)$ as follows 
\begin{align}
  &\Pr \left(\hatc(n) \neq c(n), \hats(n) \neq s(n)\right)\label{eq:ber_c2_zf} \\
  &= P_{\sf e, s} (\bH) \left(1-P_{\sf e, s} (\bH)\right)+P_{\sf e, s}^2 (\bH) \calQ \left( - \frac{\sqrt{2}}{\sigma \sqrt{A_{22}(\bH)}}\right).\nonumber
\end{align}
Thus, from~\eqref{eq:ber_c1_zf} and~\eqref{eq:ber_c2_zf}, we obtain the BER of $c(n)$ in \eqref{eq:BER_t_analysis_zf}. This completes the proof.
\end{proof}

\section{Proof of Proposition~\ref{theorem:BERMMSE}}\label{proof_Thm_MMSE}
\begin{proof}
With MMSE detection, the estimated signal vector is $\tilbx(n) = \left(\bH^H \bH + \sigma^2 \bI\right)^{-1} \bH^H \by$. Define  $\gamma_1(\bH)=\bh_1^H \left( \bh_2 \bh_2^H + \sigma^2 \bI \right)^{-1} \bh_1$. The estimated RF-source signal is written as follows~\cite{TseFWC2005}
\begin{align}
  \tils(n) &= \frac{\gamma_1(\bH)}{1+\gamma_1(\bH)} s(n) + \tilu_1(n),
\end{align}
where $\tilu_1(n) \sim \calC \calN \left(0, \frac{\gamma_1(\bH)}{(1+\gamma_1(\bH))^2}\right)$. Thus, the detecting SNR for $s(n)$ is thus $\gamma_1(\bH)$. Hence, the BER of $s(n)$ is given in~\eqref{eq:BER_s_analysis_mmse}.

Define  $\gamma_2(\bH)=\bh_2^H \left( \bh_1 \bh_1^H + \sigma^2 \bI \right)^{-1} \bh_2$.  From~\cite{TseFWC2005}, the estimated backscattered signal is
\begin{align}
  \tilx_2(n) &= s(n) c(n) + \tilu_2(n),
\end{align}
where $\tilu_2(n) \sim \calC\calN\left(0, \frac{\gamma_2(\bH)}{(1+\gamma_2(\bH))^2}\right)$. Denote $\tilu^{\sharp}_2(n)=\frac{\tilu(n)}{\hats(n)}=\tilu^{\sharp}_{2, R}(n)+j\tilu^{\sharp}_{2,I}(n)$, and $\tilu^{\sharp}_2(n) \sim \calC\calN\left(0, \frac{\gamma_2(\bH)}{2 (1+\gamma_2(\bH))^2}\right)$. Define the following signal
\begin{align}
  z(n) \!= \! \text{Re} \left\{ \! \frac{\tilx_2(n)}{\hats(n)} \! \right\} \!=\! \text{Re} \left\{ \! \frac{s(n)}{\hats(n)} \! \right\} c(n)  \!+\! \tilu^{\sharp}_{2, R}(n),
\end{align}
where $\tilu^{\sharp}_{2, R}(n) \sim \calC\calN\left(0, \frac{\gamma_2(\bH)}{2 (1+\gamma_2(\bH))^2}\right)$.

For the case $\hats(n)=s(n)$, we have $z(n)=c(n)+\tilu^{\sharp}_{2, R}(n)$. Following similar steps as in the proof of Theorem 1, we obtain the BER of $c(n)$ for the case $\hats(n)=s(n)$ as follows
\begin{align}
  &\Pr \left(\hatc(n) \neq c(n), \hats(n) = s(n)\right)\nonumber \\
  &\!=\! \left(1\!-\!P_{\sf e, s} (\bH)\right)^2  \! \Pr(c(n)\!=\!1) \! \Pr(z(n)\!<\!0 | c(n)\!=\!1) \!+\!... \nonumber \\
  &\quad \left(1\!-\! P_{\sf e, s} (\bH)\right)^2 \Pr(c(n) \!=\! 1) \Pr(z(n) \!>\! 0 | c(n) \!=\! -1) \nonumber \\
  &= \left(1\!-\!P_{\sf e, s} (\bH)\right)^2 \calQ \left( \sqrt{\bh_2^H \left( \bh_1 \bh_1^H \!+\! \sigma^2 \bI \right)^{-1} \bh_2}\right). \label{eq:ber_c1_mmse}
\end{align}

If $\hats(n) \neq s(n)$, the fraction $\frac{s(n)}{\hats (n)}$ takes the value of $j, \ -j$ and $-1$, termed as case $i =1, 2, 3$, with probability of $\sqrt{2}-1$, $\sqrt{2}-1$ and $(\sqrt{2}-1)^2$, respectively. Clearly, for each case $i$, the $z(n)$ follows Gaussian distribution with variance $\frac{\sigma^2}{2 \| \bh_2\|^2}$ and mean $\mu_1 = \mu_2 = 0$, and $\mu_3 = -c(n)$. Following similar steps as in the proof of Theorem 1, we obtain the BER of $c(n)$ for the case $\hats(n) \neq s(n)$ as follows
\begin{align}
  &\Pr \left(\hatc(n) \! \neq \! c(n), \hats(n) \! \neq \! s(n) \! \right) \!=\! P_{\sf e, s} (\bH) \! \left(1 \!-\! P_{\sf e, s} (\bH)\right)\!+\!...\nonumber \\
  &\quad \quad \quad P_{\sf e, s}^2 (\bH) \calQ \left( \!-\! \sqrt{\bh_2^H \left( \bh_1 \bh_1^H \!+\! \sigma^2 \bI_2 \right)^{-1} \bh_2}\right).  \label{eq:ber_c2_mmse}
\end{align}
Thus, from~\eqref{eq:ber_c1_mmse} and~\eqref{eq:ber_c2_mmse}, we obtain the BER of $c(n)$ in \eqref{eq:BER_t_analysis_mmse}. This completes the proof.
\end{proof}%
\bibliography{IEEEabrv,reference1705}
\bibliographystyle{IEEEtran}
\end{document}